\documentclass[%
 reprint, nofootinbib
]{revtex4-2}

\usepackage{amssymb}
\usepackage{amsthm}
\usepackage{bm}
\usepackage{graphicx}
\usepackage[colorlinks=true,linkcolor=blue,citecolor=blue,urlcolor=blue]{hyperref}
\usepackage{tikz}
\usepackage{physics}

\theoremstyle{plain}
\newtheorem{theorem}{Theorem}
\newtheorem{proposition}{Proposition}
\newtheorem{lemma}{Lemma}

\theoremstyle{definition}
\newtheorem{definition}{Definition}

\theoremstyle{remark}

\newcommand{\ak}{\hat{a}(k)}
\newcommand{\nnorm}{\frac{1}{\sqrt{N}}}
\newcommand{\psucc}{p_{\rm succ}}
\newcommand{\pstar}{p^\star}
\newcommand{\pg}{\hat{p}_g}
\newcommand{\wm}{\langle W_m\rangle_{\rho_m}}

\begin{document}
\title{Classical Limits of Spectral Filtering in Quantum Generative Models}

\author{Marco Roth}
\affiliation{Fraunhofer Institute for Industrial Engineering IAO, Nobelstr. 12, D-70569 Stuttgart, Germany}
\email{marco.roth@iao.fraunhofer.de}
\date{\today}

\begin{abstract}
Spectral filtering has been proposed as a route to regularization in quantum generative models:
the quantum Fourier transform exposes the amplitude spectrum of a quantum circuit Born machine,
and a diagonal filter suppresses the high frequencies associated with finite-sample noise,
an operation whose classical counterpart seemingly requires manipulating an exponentially long
amplitude vector. We examine whether this coherent operation produces anything that classical
post-processing of samples from the unfiltered model cannot match. Measuring the filter against
convolution with a symmetric probability kernel at matched sampling cost, which accounts for
the post-selection overhead of attenuation, we derive necessary and sufficient conditions
for the gap between the two to vanish. Magnitude (attenuating) filters obey a dichotomy: at a fixed affordability
threshold, the filtered output is either a constant-size Fourier object with an efficient
classical sampler, or the passband must widen until no fixed frequency is attenuated and the
filter no longer smooths. In neither case does the filter create a quantum-classical separation.
Whatever separation survives is inherited from the spectral phase of the input state.
Numerical experiments on trained circuit Born machines confirm the classification and show
that the deciding phases are invisible to the Born-rule training loss
and set by the initialization. Within the diagonal family, pure phase filters remain the only
spectral operations exempt from these constraints.
\end{abstract}

\maketitle

\section{Introduction}

Using quantum computers to perform machine learning tasks is an intriguing idea that has received increasing attention in recent years~\cite{Wang_2024, ji2026quantumdeeplearningcomprehensive, Cerezo2022, PhysRevApplied.21.067001}. Motivated by the potential beyond-classical
capabilities that an exponentially sized Hilbert space offers, a lot of focus has been on trying to harness the non-classical properties of such a state space in terms
of expressivity and representability~\cite{PhysRevA.103.032430, PhysRevLett.122.040504, Caro2022}. While some clear signatures of quantum advantage have been identified~\cite{Sweke2021quantumversus, doi:10.1126/science.abn7293, Liu2021, zhao2026exponentialquantumadvantageprocessing}, the vastness of the Hilbert space and its associated objects has
proven to be a double-edged sword, and the very same properties that spawned the original interest lead to a curse-of-dimensionality problem that obstructs the training of
many quantum machine learning (QML) models of interest~\cite{thanasilp2024exponentialconcentrationquantumkernel, Ragone2024, Larocca2025}.

Recently, a line of research has focused on exploiting properties beyond expressiveness as a road to a potential quantum advantage. In this context, the quantum Fourier transform (QFT)
has received renewed attention as a tool to design ML models~\cite{wakeham2024inferenceinterferenceinvariancequantum, simidzija2026solvingapproximatehiddensubgroup, coffman2026groupfourierfilteringquantum}. It has been argued that the QFT enables spectral design operations whose classical counterparts appear to require
manipulating an exponentially long amplitude vector and are therefore suggested to be classically prohibitive~\cite{belis2026spectralmethodscrucialmachine}. Viewed through the lens of Fourier analysis, the simplicity
biases that underpin generalization acquire a concrete spectral form: smooth distributions have decaying Fourier spectra, so suppressing the high-frequency content of an empirical 
distribution acts as a regularizer that discards finite-sample artifacts [cf. Fig.~\ref{fig:overview} (a) for an illustration]. The natural model class for such spectral design is the quantum circuit Born
machine~\cite{PhysRevA.98.062324, Benedetti_2019}. In these models, a quantum state is prepared by a
(typically parameterized) circuit that defines an implicit generative model
through the Born rule from which one can draw samples but
not evaluate likelihoods. Born machines are among the leading candidates for a
quantum advantage in generative modeling, since sampling from generic circuit
families is believed to be classically
intractable~\cite{Coyle2020, PhysRevX.12.021037, Rudolph2024}. Notably, the output
distribution depends only on the magnitudes, while the phases are invisible to it. This is the basis of a recent
proposal to design quantum generative models by filtering their spectrum coherently, before measurement~\cite{belis2026spectralmethodscrucialmachine}. 

In this work, we explore the necessary conditions for such an operation to bear a potential quantum advantage for generative quantum models. A key insight is that a classical
competitor never needs the amplitude vector, only the measured distribution,
to which it has sampling access. We make the comparison precise for diagonal spectral filters acting on the Fourier-transformed state. We measure the coherent
filter against a classical competitor with the same interface, namely post-processing of samples from the unfiltered model by convolution with a probability kernel.
We take the operational reading throughout. Two models are the same if no
sampler can tell them apart at the same sampling cost, and a filter earns its coherence only if it
produces an output that no classical post-processing of the unfiltered samples can match. 

\begin{figure*}[t]
    \centering
    \includegraphics[width=0.76\textwidth]{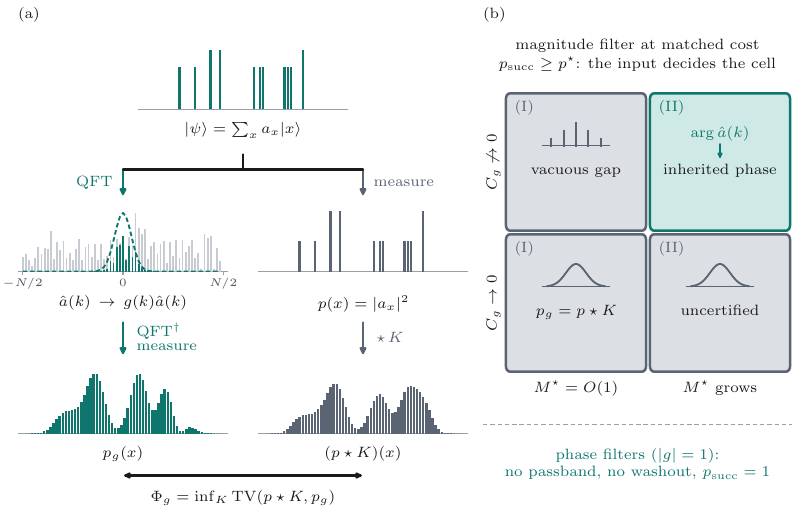}
    \caption{Overview of this work.
    (a) The two routes from the same quantum generative model Eq.~\eqref{eq:quantum_born_machine}.
    Coherent route (left): a QFT exposes the amplitude spectrum, a spectral filter attenuates it, $\ak\mapsto g(k)\ak$ (dashed: Gaussian low-pass at matched cost $\psucc\geq\pstar$), and measuring after the inverse QFT yields the filtered distribution $p_g$.
    Classical route (right): measuring first yields samples of $p(x)=\abs{a_x}^2$, which are convolved with a classical kernel $K$.
    The gap $\Phi_g$, Eq.~\eqref{eq:gap}, measures what the coherent route achieves beyond classical smoothing. The example shows the optimal kernel ($\mathrm{TV}=0.17$).
    (b) The resulting classification. At matched cost, a magnitude filter
    ($\abs{g(k)}\leq1$) either returns an $O(1)$-mode object with an explicit kernel surrogate,
    or leaves a separation that is inherited from the spectral phase of the input rather than
    produced by the filter. Genuinely quantum spectral operations are phase filters ($\abs{g(k)}=1$), which incur no sampling cost but admit no kernel surrogate.}
    \label{fig:overview}
\end{figure*}

We quantify the difference by a single quantity: the smallest total-variation distance between the coherently filtered distribution and any such classical
smoothing. We prove the conditions under which a coherent filter operation is equivalent to classical smoothing.
More generally, we find that it is instructive to separate \emph{magnitude} (attenuating) filters
from \emph{phase} filters in less restricted settings. We show that magnitude filters are fundamentally restricted by a dichotomy. Either the filtered output is
a constant-size Fourier object, or any quantum-classical separation is traceable to spectral phase already present in the
input state. In the former case, the output distribution can be reproduced by a classical model directly. In the latter case, the filter merely windows at a post-selection cost.
As a consequence, we find that a magnitude filter cannot perform smoothing at an affordable cost while leaving the spectral separation intact. Phase filters are the only spectral operations left that are genuinely quantum. The classification
is illustrated in Fig.~\ref{fig:overview} (b).

In the remainder of this work, Sec.~\ref{sec:framework} sets up quantum generative models and coherent spectral filtering, and Sec.~\ref{sec:results} develops a
classicality criterion based on the coherent residual, the compactness of the filtered output, and the resulting dichotomy for magnitude filters.
Section~\ref{sec:examples} illustrates the two cells with spectrally incoherent input, a Gaussian low-pass, and numerical experiments on random and trained Born machines.
We conclude in Sec.~\ref{sec:discussion}. Proofs and technical details are deferred to the appendices.

\section{Framework}
\label{sec:framework}
In this section, we briefly introduce the concepts used throughout the manuscript. We start by explaining quantum generative models as the object of interest
and discuss how quantum operations that shape the amplitudes of such models can be used to design the output distributions.
\subsection{Quantum Generative Models}
\label{sec:quantum_generative_models}
Consider generative machine learning models, i.e., models that allow drawing samples from a distribution $p(x)$, where the models are trained to mimic a target distribution
$p_{\rm target}$. Creating models that \emph{learn} the underlying distribution instead of overfitting to particular examples is one of the core problems in machine learning.
Fundamentally, this results in a bias-variance trade-off~\cite{Hastie2009}. A standard approach to addressing this trade-off is to introduce a simplicity bias into the model by artificially
constraining the hypothesis space~\cite{kubler2021inductive}.

In the following, we focus on achieving this regularization through manipulation of a model in \emph{Fourier} space rather than operating on the feature space directly.
This is motivated by the following insight: in practical scenarios, we can only approximate the true target distribution $p_{\rm target}$ using an empirical distribution that
is obtained by collecting samples. The empirical distribution of a finite training set is typically sparse, concentrating all of its
mass on the observed samples. This results in a Fourier spectrum that
spreads across the entire frequency range, with appreciable weight on high-order components
that reflect finite-sample noise rather than some underlying structure. Suppressing these high
frequencies with a low-pass filter therefore acts as a smoothness prior by discarding the
data-specific detail and returning a distribution that interpolates between the training
points~\cite{belis2026spectralmethodscrucialmachine}. 

For quantum machine learning models, this is often taken to be promising, since the QFT allows for efficient access to an exponentially
sized amplitude vector.
Concretely, we consider a generative quantum machine learning model that is obtained by preparing an $n$-qubit quantum state 
\begin{equation}
  \ket{\psi(\theta)}=\sum_{x=0}^{N-1} a_x(\theta)\ket{x}=U(\theta)\ket{0}\,,\label{eq:quantum_born_machine}
\end{equation}
using a parametrized quantum circuit $U(\theta)$. Here $a_x\in\mathbb{C}$ are the amplitudes of $\ket{\psi}$ in the basis $\lbrace\ket{x}\rbrace$ with $x\in\lbrace 0,\dots,N-1\rbrace$ and $N=2^n$ is the number of computational basis states of $\ket{\psi}$. From Eq.~\eqref{eq:quantum_born_machine},
we obtain the model distribution through the Born rule $p(x;\theta)=\abs{a_x(\theta)}^2$. The resulting models are thus called quantum Born machines~\cite{PhysRevA.98.062324}. Learning the target distribution can be operationalized
as minimizing a suitable loss function such as the forward Kullback--Leibler divergence $\mathcal{L}(\theta)=D_{\rm KL}(p_{\rm target}\,\|\,p(\cdot\,;\theta))$. This is achieved by optimizing the parameters $\theta$, e.g., using gradient descent.

In this work, we assume access to a (trained) Born machine \emph{before} measurement such that we can perform operations on the state of Eq.~\eqref{eq:quantum_born_machine}
directly. This will be made precise in the subsequent section.

\subsection{Spectral Filtering}
\label{sec:spectral_filtering}
In the following, we study spectral operations on the quantum Born machine. To this end, we apply a QFT to the state of Eq.~\eqref{eq:quantum_born_machine}, which yields the transformed amplitudes
\begin{equation}
    \ak = \nnorm\sum_x a_x\omega^{-kx}\,,
    \label{eq:fourier_transf}
\end{equation}
where $\omega\equiv\exp(i2\pi/N)$. Here we have omitted the dependence on the parameters $\theta$ for clarity. Note that for distributions we use the unnormalized transform
$\hat f(m)=\sum_x f(x)\omega^{-mx}$, so that $\hat p(0)=1$.
The $1/\sqrt N$ in Eq.~\eqref{eq:fourier_transf} is reserved
for amplitudes, where the transform must be unitary.

Two different spectra are now in play, and it is worth separating them at the outset. The
regularization argument of Sec.~\ref{sec:quantum_generative_models} is about the spectrum of the
\emph{distribution}, $\hat p(m)$, whereas the QFT of Eq.~\eqref{eq:fourier_transf} gives access to the
spectrum of the \emph{amplitudes}, $\ak$. The two are related by the Wiener--Khinchin theorem, which
in the above conventions is free of prefactors,
\begin{equation}
    \hat p(m) = \sum_k \rho_m(k)\,,
    \label{eq:wiener_khinchin}
\end{equation}
with the lag-$m$ amplitude product $\rho_m(k) \equiv \hat a(k)\hat a^*(k-m)$. The spectrum
of the distribution is thus the autocorrelation of the spectrum of the amplitudes, with the
lag $m$ playing the role of the shift between the two copies. For $m=0$ this returns the
normalization $\hat p(0)=\sum_k\abs{\ak}^2=1$. 

We are interested in delineating two scenarios: coherent manipulation of the amplitudes of Eq.~\eqref{eq:fourier_transf}
before measurement [left track in Fig.~\ref{fig:overview} (a)] and classical post-processing operating on the output distribution $p(x)$ obtained through 
sampling [right track in Fig.~\ref{fig:overview} (a)]. To decide whether the additional freedom of coherent access has any operational consequences, we must weigh it against what an
 experimenter with the same interface, i.e., sampling access to $p(x)$, could already achieve. We thus first fix that
classical baseline and then turn to its coherent counterpart.

To realize the regularization motivated in Sec.~\ref{sec:quantum_generative_models}, we
 model classical smoothing by linear, probability-preserving maps $G_K$ that are
invariant under translation and reflection. The latter two symmetries define a smoothing filter without shifting.
We argue why restricting the classical competition model to this class is appropriate in Appendix~\ref{app:framework}.
Under these requirements $G_K$ acts on the output distribution $p(x)$ as a
convolution with a symmetric probability kernel $K$, $G_K:\ p\mapsto p\star K$~\cite{oppenheim}, with a real, even transfer function $\hat K(m)$. By the convolution theorem this acts on the spectrum of $p$ as a
pointwise multiplication, so the filtered distribution has Fourier coefficients
\begin{equation}
    \hat{p}_G(m) \;=\; \hat{K}(m)\,\hat{p}(m)\,,
    \label{eq:convolution_filter}
\end{equation}
i.e., $G_K$ rescales each mode $m$ by the real, even factor $\hat{K}(m)$. 

Equation~\eqref{eq:convolution_filter} is the result of classical post-processing. The coherent quantum analogues of convolution filters on the
output distribution as given in Eq.~\eqref{eq:convolution_filter} are spectral filters $g:\{-N/2,\dots,N/2-1\}\to\mathbb{C}$ that are diagonal
in the Fourier basis, i.e., $g: \ak\mapsto g(k)\ak$. In the following, we assume a normalized filter $\norm{g}_\infty=1$.
As shown in Appendix~\ref{app:framework}, the modes obtained after applying a filter, followed by an inverse QFT and subsequent measurement of the amplitudes,
can be factorized as
\begin{equation}
    \pg(m)=\frac{1}{\psucc}\sum_k W_m(k)\rho_m(k)\,,
    \label{eq:master_id}
\end{equation}
where $W_m(k) \equiv g(k)g^*(k-m)$, and 
\begin{equation}
    \psucc = \sum_k\abs{g(k)}^2\abs{\ak}^2\,.
    \label{eq:psucc}
\end{equation}

Equation~\eqref{eq:master_id} reveals that the filter modifies the distribution via pairwise
spectral weights $W_m(k)$, the lag-$m$ \emph{overlap} of the filter with its own shifted copy.
These spectral weights do not act on the amplitudes directly but rather on their autocorrelation
$\rho_m(k)$.
The implications can be understood by reading Eq.~\eqref{eq:wiener_khinchin} against Eq.~\eqref{eq:master_id}:
$\rho_m(k)$ is the summand of the Wiener--Khinchin sum, and the
filter inserts the weight $W_m(k)$ \emph{inside} that sum. Because the insertion happens under the
sum and not in front of it, the induced action on $\hat p$ is an average of $W_m$ over the terms of the
autocorrelation, rather than the pointwise multiplication of Eq.~\eqref{eq:convolution_filter}.

In the following, we
call a lag \emph{regular} if $\hat{p}(m)\neq 0$ and \emph{degenerate} if $\hat{p}(m)= 0$, and we collect
the regular lags in the set
\begin{equation}
    \mathcal{S} \equiv \lbrace m : \hat{p}(m)\neq 0\rbrace\,.
    \label{eq:regular_lags}
\end{equation}
We can use the previous definitions to write down the coherent equivalent of Eq.~\eqref{eq:convolution_filter}
for regular lags $m\in\mathcal{S}$
\begin{equation}
    \hat{p}_g(m) = \frac{\langle W_m\rangle_{\rho_m}}{\psucc}\hat{p}(m)\,.
    \label{eq:coherent_filter}
\end{equation}
Here, $\wm=\sum_k W_m(k)\rho_m(k)/(\sum_k\rho_m(k))=\sum_k [W_m(k)\rho_m(k)]/\hat{p}(m)$ is the $\rho_m$-weighted average
of the pair weight $W_m(k)$, the second equality being Eq.~\eqref{eq:wiener_khinchin}. Comparing Eq.~\eqref{eq:convolution_filter} to Eq.~\eqref{eq:coherent_filter} reveals two important distinctions between coherent manipulation of the
spectrum and filtering after measurement in classical post-processing: classical post-processing modifies the spectrum by a multiplication with a \emph{real} scalar, whereas
$\wm$ in Eq.~\eqref{eq:master_id} is \emph{complex}, which allows for additional modification of the phases of $\hat{p}(m)$.
The appearance of the normalization $\psucc$ in Eq.~\eqref{eq:psucc} has an operational interpretation: spectral attenuation (e.g., a low-pass filter)
requires $\abs{g(k)}<1$ for some $k$. This is a non-unitary operation that requires ancillary measurements with post-selection
on measurement outcomes with success probability $\psucc$ (see Appendix~\ref{app:framework}). As a consequence, compared to sampling the non-filtered distribution with a given measurement
budget of $N_{\rm shots}$, filtering with $\abs{g(k)}<1$ and subsequent sampling results in only $\psucc N_{\rm shots}$ effective samples. Note that more sophisticated schemes such as
 amplitude amplification~\cite{lomonaco2002quantum} could increase the efficiency of the post-selection. To make the following analysis independent of the specific post-selection method,
 we define an affordability threshold $\pstar$ in the following discussions (cf. Sec.~\ref{sec:compact}). 

\section{Results}
\label{sec:results}

These two differences, a complex rather than real reweighting and a post-selection
cost, are the entire gap between the coherent filter and its classical counterpart. The
rest of the paper is intended to make the implications of this difference concrete: first, whether the complex reweighting ever
produces an output that no classical smoothing can reach, and second, whether any
such advantage survives once its cost and spectral complexity are accounted for.

\subsection{A Classicality Criterion}

\begin{figure*}[t]
    \centering
    \includegraphics[width=\textwidth]{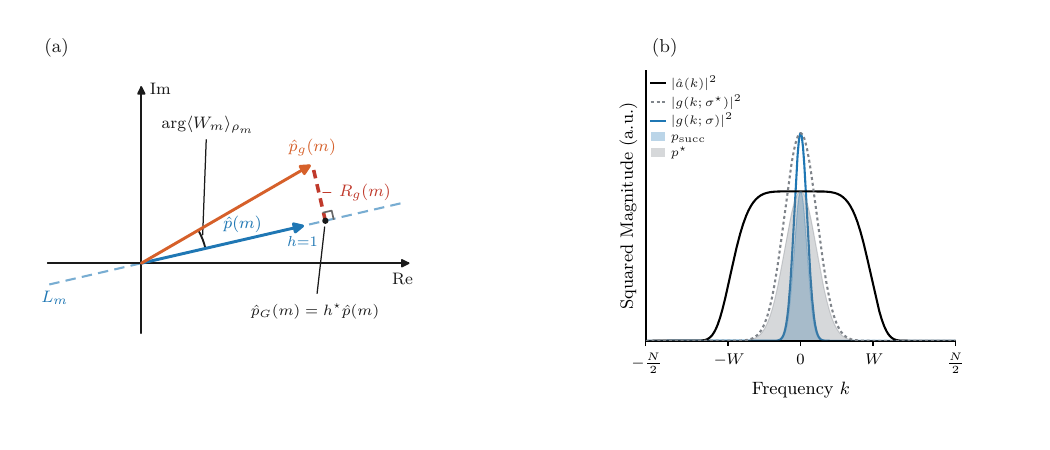}
    \caption{(a) Geometry of the coherent residual $R_g$ [Eq.~\eqref{eq:coherent_residual}] at a fixed lag $m$. For a given input distribution $\hat{p}(m)$,
    the coherent filter can rotate the modes $p_g(m)$ [Eq.~\eqref{eq:master_id}], while classical smoothing can only rescale each mode by the real factor $\hat{K}(m)$,
     resulting in $\hat{p}_G(m)$ [Eq.~\eqref{eq:convolution_filter}]. All
    modes reachable by classical smoothing are thus on a line $L_m=\lbrace h\hat{p}(m):h\in\mathbb{R}\rbrace$ through the origin.
    The best classical value $h^\star(m)\hat{p}(m)$ is the orthogonal projection of $\hat{p}_g(m)$ onto the line $L_m$. The coherent residual $R_g$ is the perpendicular component, which cannot be removed by
    a kernel. 
    The rotation of $\hat{p}_g(m)$ off $L_m$ is set by $\arg\langle W_m\rangle_{\rho_m}$, with $R_g(m)\propto\mathrm{Im}\langle W_m\rangle_{\rho_m}$.
    (b) Relationship between the filter Fourier transformed modes $\hat{a}(k)$, the filter $g(k;\sigma)$ and the relevant probabilities $\psucc$ and $p^\star$.
    In this example, we have chosen a low-pass filter $|g(k;\sigma)|^2=\exp(-k^2/\sigma^2)$ (blue, solid) applied to a flat-topped spectrum of half-width $W$, (black). Filtering the spectrum retains the fraction
    $p_{\mathrm{succ}}(\sigma)$ [Eq.~\eqref{eq:psucc}] of the unit
    spectral mass $|\hat a(k)|^2$, drawn as the shaded area under the window (blue).
    The affordable width $\sigma^\star$ [Eq.~\eqref{eq:sigmastar}] is obtained as the smallest $\sigma$ with
    $p_{\mathrm{succ}}(\sigma^\star)=p^\star$. The resulting filter is shown
    in grey (dashed). Its window then captures exactly $p^\star$
    (grey area). In the shown example $\sigma^\star>\sigma$, such that $p_{\mathrm{succ}}<\pstar$ such that the filter is classified as unaffordable.}
    \label{fig:simulability_criterion}
\end{figure*}

To determine whether the additional freedom granted by the coherent filter has operational
consequences, we introduce the gap
\begin{equation}
    \Phi_g \;\equiv\; \inf_{K\in\mathcal{K}}\ \text{TV}\!\big(p\star K,\ p_g\big)\,,
    \label{eq:gap}
\end{equation}
the smallest total-variation distance, $\text{TV}(z,z')=\tfrac12\sum_x\abs{z(x)-z'(x)}$, between the
coherently filtered distribution $\pg$ and any classical smoothing $G_K$ of the unfiltered
samples. The minimization is carried out over the space $\mathcal{K}$ of symmetric probability kernels. In the following, we assume a large sampling limit
in which shot noise is not a dominant factor that decides whether $\Phi_g$ vanishes or not. In this limit,
we have $\Phi_g=0$ exactly when classical smoothing reproduces the coherent output. In this case
the filter confers no advantage. A positive gap, however, witnesses an output no classical smoothing of
the samples can reach. 

To identify the regimes with $\Phi_g > 0$, we consider the coherent residual
\begin{equation}
    R_g(m) = \hat{p}_g(m)-h^\star(m)\hat{p}(m)\,,
    \label{eq:coherent_residual}
\end{equation}
i.e., the difference between the modes obtained through a coherent filter and classical post-processing. Here $\Phi_g$ is the operational quantity, deciding whether a classical
smoothing exists, while $R_g(m)$ is the per-lag diagnostic that certifies a
gap and attributes it to the input phases. A nonvanishing residual is
sufficient but not necessary for $\Phi_g>0$, and a residual vanishing at
every lag is necessary but not sufficient for $\Phi_g=0$.
Here, we define $h^\star(m)=\arg\min_{h\in\mathbb{R}}\abs{\hat{p}_g(m)-h\hat{p}(m)}^2$ on regular lags $m\in\mathcal{S}$
as the classical weight
that minimizes the squared difference between the coherently filtered mode and the mode obtained
through classical post-processing. On degenerate lags, we define $h^\star(m)\hat{p}(m)\equiv 0$. As derived in Appendix~\ref{app:criterion}, the minimum
can be computed in closed form, which for $m\in\mathcal{S}$ yields
\begin{equation}
    \label{eq:residual_final_form}
      R_g(m)=i\,\mathrm{Im}\langle W_m\rangle_{\rho_m}\,\frac{\hat p(m)}{\psucc}\,.
\end{equation}
From this we determine two cases where the coherent residual can vanish.
First, deterministically, when the $\rho_m$-weighted average $\langle W_m\rangle_{\rho_m}$ is real. Second, statistically, when the input is spectrally incoherent: there $\text{Im}\wm$ concentrates at zero over the random phases and the phase-charged residual vanishes lag by lag in mean square at fixed affordability
$\psucc \geq \pstar$. We formalize the first case in the following proposition and return to the second case later.
\begin{proposition}[Simulability criterion]\label{prop:criterion}
The coherent filter is exactly classical smoothing and $\Phi_g=0$, if and only if two independent
conditions both hold:
\begin{enumerate}
  \item[(a)] \emph{Phase alignment.} $R_g(m)=0$ at every lag. Equivalently $\wm\in\mathbb{R}$ at
every regular lag and $\pg(m)=0$ at every degenerate lag.
  \item[(b)] \emph{Envelope realizability.} The real envelope
  $h^\star(m)=\mathrm{Re}\langle W_m\rangle_{\rho_m}/\psucc$ is the transfer function of a symmetric probability kernel for
  regular lags, i.e. there exists $K\geq 0$ with
  $\sum_x K(x)=1$ and $\hat{K}(m)=h^\star(m)\,,\forall m\in\mathcal{S}$.
\end{enumerate}
\end{proposition}

Proposition~\ref{prop:criterion} locates the entire quantum content of a
spectral filter in a single scalar per lag, the $\rho_m$-averaged pair-weight
$\langle W_m\rangle_{\rho_m}$. The two conditions constrain different parts of $\wm$. Condition (a) constrains its \emph{imaginary} part: the classical side can rescale
$\hat p(m)$ but not rotate it, so by Eq.~\eqref{eq:residual_final_form}, the residual is the component of $\hat p_g(m)$ orthogonal to $\hat p(m)$, carried entirely by the input phases
through $\rho_m$.  Figure~\ref{fig:simulability_criterion} (a) visualizes the relationship between the filtered distributions and the coherent residual.

Condition (b) constrains its \emph{real}
part: a real even envelope normalized to $h^\star(0)=1$ need not arise from any nonnegative kernel,
and when it does not, $\Phi_g>0$ even though (a) holds and every residual vanishes. This is a purely
classical realizability constraint with no dependence on the phases. The elementary classical case of the low-pass filter (idealized constant on $\text{supp}\rho_m$)
satisfies both: a $\rho_m$-constant pair-weight $W_m(k)\equiv w(m)\in\mathbb{R}$ on the support of $\rho_m$ aligns every lag, and the
canonical low-pass profiles induce a positive-definite $h^\star$.

To turn Prop.~\ref{prop:criterion} into a classification of actual filters, we write the
pair-weight in polar form using $g(k)=\abs{g(k)}e^{i\theta(k)}$, which results in
\begin{equation}
    W_m(k)=\abs{g(k)}\,\abs{g(k-m)}\;e^{i[\theta(k)-\theta(k-m)]}\,.
    \label{eq:W_polar}
\end{equation}
Based on this decomposition, we define two archetypes: \emph{magnitude} filters ($\theta\equiv0$),
whose pair-weight is real, and \emph{phase} filters ($\abs{g}\equiv1$), whose
pair-weight has unit modulus.

A positive gap establishes separation from smoothing but not from classical simulation. Additionally, a useful quantum filter should be affordable and its output
should not be captured by, e.g., efficient Fourier analysis. In the following, we discuss both the affordability and the
spectral complexity of magnitude and phase filters.

\subsection{Affordability and Spectral Compactness}
\label{sec:compact}

To make the following discussion more concrete, we focus on filters that straightforwardly
achieve a high-frequency regularization as discussed in Sec.~\ref{sec:spectral_filtering}: one-parameter \emph{magnitude-filter families}
$g_\sigma$, with low-pass profiles in which a single scale $\sigma>0$ sets how
aggressively high frequencies are attenuated. We call $\sigma$ the
\emph{passband scale}, where the \emph{passband} is the band of frequencies
transmitted with appreciable weight. We have $\abs{g_\sigma(k)}$  of order unity for
$\abs{k}\lesssim\sigma$ and $|g_\sigma(k)|$ decaying for $|k|\gtrsim\sigma$. As discussed in Sec.~\ref{sec:spectral_filtering}, magnitude filters with $\abs{g(k)}<1$ result in a reduced success probability $\psucc < 1$.\footnote{In contrast, pure phase filters with unit norm $\abs{g(k)}=1$, $\psucc=1$ such that the filtering operation does not incur an additional sampling cost.}
As this implicitly imposes a cost on the number of measurements needed to create samples, we compare filters only at a fixed
affordability floor $\psucc\ge\pstar$, in the following denoted as \emph{matched-cost constraint}, which puts the quantum
filter and its classical competitor on the same sampling budget. For the filtering families $g_\sigma$, we define 
the \emph{affordable passband scale}
\begin{equation}\label{eq:sigmastar}
  \sigma^\star(n)\;=\;\min\{\,\sigma\ \colon\ \psucc(\sigma)\ge\pstar\,\}\,.
\end{equation}
It is the smallest $\sigma$ with $\psucc\ge \pstar$. A visualization of the involved quantities is shown in Fig.~\ref{fig:simulability_criterion} (b).

Given a family of filters $g_\sigma$ as defined above, we now show that the mode count of the filtered output is set by
the passband scale and not by the grid size $N=2^n$. Whether this results in a
classically substitutable object is decided by how fast $\sigma$ must grow with $n$ for the filter to remain
affordable. To this end, we define the smallest number of modes required to reproduce $\pg$ to accuracy $\epsilon$
as 
\begin{equation}\label{eq:modecount}
  M^\star(p_g;\epsilon)\;=\;\min\bigl\{\,M\ \colon\
  \mathrm{TV}\bigl(p_g,\,p_g^{\le M}\bigr)\le\epsilon\,\bigr\}\,.
\end{equation}
Here, we have introduced the filtered output truncated to its $2M+1$ lowest frequencies
\begin{equation}\label{eq:truncation}
  p_g^{\le M}(x)\;=\;\frac{1}{N}\sum_{\abs{m}\le M}\hat p_g(m)\,\omega^{mx}\,.
\end{equation}
Using these definitions, we can state a bound on $M^\star$ formalized in the following Lemma.
\begin{lemma}[Spectral compactness collapse]\label{lem:compact}
Let $\{g_\sigma\}$ be a family with passband scale $\sigma$, i.e., a filter family whose
overlap obeys $\sup_k\big(\abs{g_\sigma(k)}\,\abs{g_\sigma(k-m)}\big)\le
s(\abs{m}/\sigma)$ for all $m$, with a fixed, non-increasing, square-integrable overlap profile
$s:[0,\infty)\to[0,1]$. Then for
every input state with $\psucc \geq \pstar$, the filtered output $p_g$ can be truncated to its $\abs{m}\le
M^\star$ Fourier modes with total-variation error at most $\epsilon$, where
\begin{equation*}
  M^\star(p_g;\epsilon)\;\le\;\sigma\,T^{-1}\!\Big(\frac{2(\pstar\epsilon)^2}{\sigma}\Big),
  \qquad T(a):=\int_a^\infty s(u)^2\,du,
\end{equation*}
is a bound independent of the grid size $N=2^n$.
\end{lemma}
The proof is in Appendix~\ref{app:compact} and relies on bounding $\abs{\pg(m)}$ by
the supremum of the pair-weight $W_m(k)$ using Cauchy-Schwarz followed by converting the high-lag
tail into a TV bound using Parseval's theorem. The bound for $M^\star$ is made explicit for the example
of a Gaussian low-pass filter in Sec.~\ref{sec:gaussian}. 

How many modes are needed to approximate the filtered
distribution to an error $\epsilon$ is governed entirely by how fast the maximum pair weight decays with the lag $m$. For a
family $g_\sigma$ this decay is not arbitrary: enlarging $\sigma$ simply stretches the
filter along the frequency axis, so the overlap at lag $m$ and passband scale $\sigma$ depends only
on the ratio $m/\sigma$. In the hypothesis of Lemma~\ref{lem:compact} this is captured by the overlap-profile $s$ as a single, $\sigma$-independent
quantity which is the universal shape of the overlap when lag is
measured in units of the passband scale.
This factors out $\sigma$ and isolates the scale from the shape.

From Lemma~\ref{lem:compact} it is apparent that a magnitude filter with a bounded passband scale ($\sigma=O(1)$) produces an
 $O(1)$-mode output uniformly in $n$, a constant-size classical object. A pure phase filter ($|g|\equiv1$) has overlap profile $s\equiv1$, which is not square-integrable. Consequently, the bound does not apply, there is no passband restriction, and the output spectrum can be non-compact.

Widening the passband raises the
success probability Eq.~\eqref{eq:psucc}, since
more spectral mass is retained, but it also widens the overlap width, i.e., the range
of lags $m$ on which $g_\sigma(k)$ and
$g_\sigma(k-m)$ are simultaneously appreciable. Affordability therefore pushes
$\sigma$ up while compactness pushes it down. We can make this tension more concrete for
the important class of broadband spectra. Call a family of inputs indexed by $n$
\emph{broadband} if its spectral mass escapes every fixed window, i.e.
\begin{equation}\label{eq:broadband}
  \sum_{\abs{k}\le L}\abs{\ak}^2\;\to\;0\quad(n\to\infty)\,,
\end{equation}
for each fixed L. Splitting the sum Eq.~\eqref{eq:psucc} at $\abs{k}=L$, and using $\abs{g}\le1$ inside the window and
$\sum_k\abs{\ak}^2=1$ outside it, gives
\begin{equation}
  \psucc\;\le\;\sum_{\abs{k}\le L}\abs{\ak}^{2}
  \;+\;\sup_{\abs{k}>L}\abs{g(k)}^{2}\,.
  \label{eq:psucc_split}
\end{equation}
The first term is all the filter can collect from the fixed window, and on a broadband input it
vanishes: the retained mass must come from outside the window. Affordability $\psucc\ge\pstar$
therefore forces
\begin{equation}
  \sup_{\abs{k}>L}\abs{g(k)}^{2}\;\ge\;\pstar-o(1)\,,
  \label{eq:no_attenuation}
\end{equation}
for every fixed $L$. That is, however wide the window, some frequency beyond it is still transmitted with amplitude
$\sqrt{\pstar}$ up to $o(1)$. An affordable magnitude filter on a broadband input can thus not be
a low-pass with an $n$-independent cutoff: suppressing the whole tail $\abs{k}>L$ to $o(1)$ at
fixed $L$ sends $\psucc\to0$. For a filter family $g_\sigma$ the same conclusion is read off
$\sigma^\star(n)\to\infty$ directly, since a diverging passband scale eventually covers any fixed
window. Regularization and affordability are competing for 
broadband spectra.

\subsection{Dichotomy for Magnitude Filters}
\label{sec:dichotomy}
Two independent mechanisms are now in place: Lemma~\ref{lem:compact} bounds the \emph{size} of
the filtered output through $\abs{g}$ alone, and the coherent residual Eq.~\eqref{eq:residual_final_form} fixes the
\emph{direction} in which that output leaves the classical set. In this section, we combine these two
with the affordability threshold set by the affordable scale $\sigma^\star(n)$ of Eq.~\eqref{eq:sigmastar},
to establish a classification for magnitude filters.

To this end, it is useful to observe that the gap $\Phi_g$ is lower bounded by the coherent residual as
\begin{equation}
  \Phi_g\;\ge\;C_g\;\equiv\;\tfrac12\max_{m\in\mathcal S}\abs{R_g(m)}\,.
  \label{eq:gap_lower}
\end{equation}
We call the quantity $C_g$ the \emph{phase certificate}, the derivation is shown in (Appendix~\ref{app:dichotomy}) and requires no
hypothesis on $g$, on the input, or on the realizability of any envelope.

For a magnitude filter $W_m$ is real, so by Eq.~\eqref{eq:residual_final_form} the certificate is
nonzero only if $\rho_m(k)=\ak\hat a^*(k-m)$ carries phase. It is therefore a statement about the
phases of the input states rather than about the filter. We thus turn our attention to
the input phases.
\begin{definition}[Input phases]\label{def:incoherent}
Fix a magnitude-filter family at matched cost. A family of inputs indexed by $n$ is
\begin{itemize}
  \item \emph{phase-trivial} if $\ak=e^{i\alpha}\omega^{tk}r_k$ for some $\alpha\in\mathbb{R}$,
  $t\in\tfrac12\mathbb{Z}_N$ and real $r_k$, i.e. if the spectrum is real up to a global phase and a
  (possibly half-integer) translation of the state. 
  \item \emph{spectrally incoherent} if $a_x=\sqrt{p(x)}\,e^{i\varphi_x}$ with $\varphi_x$
  independent and uniform.
  \item \emph{spectrally coherent} if $C_g\not\to0$.
\end{itemize}
\end{definition}
We stress that phase-triviality constrains $\arg\ak$ and not $\arg a_x$: computational-basis phase
erasure is not spectral phase erasure. Real nonnegative amplitudes only make the spectrum Hermitian,
$\hat a(-k)=\hat a(k)^*$, which is real up to a translation precisely when $p$ is
reflection-symmetric. For an asymmetric $p$ the certificate is generically nonzero even though the
state carries no computational-basis phase at all. Using these definitions, we formulate the
following proposition.
\begin{proposition}[Phases and magnitudes at matched cost]\label{prop:mechanisms}
Fix an accuracy $\epsilon>0$, an affordability threshold $\pstar\in(0,1)$, and a magnitude-filter
family $\{g_\sigma\}$ [$\theta\equiv0$ in Eq.~\eqref{eq:W_polar}] with an overlap profile in the
sense of Lemma~\ref{lem:compact}. Let $\sigma^\star(n)$ be the affordable passband scale
Eq.~\eqref{eq:sigmastar} and let $p_g$ be the output of $g_{\sigma^\star(n)}$.
\begin{itemize}
  \item[(i)] $\sigma^\star(n)$ is a functional of $\abs{\ak}$ alone, and
  the mode count bound Eq.~\eqref{eq:modecount} is 
  independent of $N=2^n$ and of the input phases. In particular $\sigma^\star(n)=O(1)$ forces
  $M^\star=O(1)$ uniformly in $n$.
  \item[(ii)] At every regular lag $m\in\mathcal S$,
  $R_g(m)\neq0\iff\mathrm{Im}\wm\neq0$, and $W_m$ being real, $\mathrm{Im}\wm$ is sourced solely by
  the phase of $\rho_m(k)$. Hence $C_g>0$ only if the input carries spectral phase, and $C_g=0$
  identically on phase-trivial inputs.
  \item[(iii)] If in addition the induced envelope
  $h^\star=\mathrm{Re}\wm/\psucc$ is realizable [condition~(b) of Prop.~\ref{prop:criterion}] and
  $\pg(m)=0$ at every degenerate lag, then
  \begin{equation}
    C_g\;\le\;\Phi_g\;\le\;\tfrac12\Big(\sum_{m\in\mathcal S}\abs{R_g(m)}^2\Big)^{\!1/2}\,,
    \label{eq:gap_sandwich}
  \end{equation}
  and in particular $\Phi_g=0$ if and only if $C_g=0$.
\end{itemize}
\end{proposition}
The three parts are proved in Appendix~\ref{app:dichotomy}: (i) is Lemma~\ref{lem:compact}
evaluated at $\sigma^\star$, (ii) follows from Eq.~\eqref{eq:residual_final_form} together with
the reality of $W_m$, and (iii) evaluates the exact $L^2$ decomposition of
Appendix~\ref{app:criterion} at the kernel realizing $h^\star$. Using this proposition, we arrive at the
following classification for magnitude filters.
\begin{theorem}[Dichotomy]\label{thm:dichotomy}
Assume the hypotheses of Prop.~\ref{prop:mechanisms}. Every family of inputs falls into
exactly one of the following two cells, according to whether the affordable passband scale
$\sigma^\star(n)$ remains bounded.
\begin{itemize}
\item[(I)] \emph{Classically reproducible.} If $\sigma^\star(n)=O(1)$, then
  $M^\star(p_g;\epsilon)=O(1)$ uniformly in $n$ and for arbitrary input phases, and $p_g$ is
  reproduced to $\mathrm{TV}\le2\epsilon$ by a distribution samplable in $\mathrm{poly}(n)$ time
  from those $O(1)$ Fourier coefficients. The model is classically simulable regardless of the
  values of $C_g$ and $\Phi_g$.
\item[(II)] \emph{No affordable cutoff.} If $\sigma^\star(n)$ is unbounded, then by
  Eq.~\eqref{eq:no_attenuation} the filter cannot implement a cutoff at any fixed frequency and
  hence does not smooth. Distinguish two cases (a) \emph{certified}, if $C_g\not\to0$, the separation
  is carried by the spectral phase of
  the input. (b) \emph{uncertified}, if $C_g\to0$ then no single-lag Fourier observable
  separates $p_g$ from classical smoothing. 
\end{itemize}
In neither cell is a separation attributable to the filter, which only contributes the real window
$W_m$ and the post-selection cost $\psucc$.
\end{theorem}
The theorem asserts that at matched cost a magnitude filter either returns a constant-size
classical object or stops smoothing altogether, and in neither case does it create the
separation, it only inherits whatever phase the input already carried. The cell is decided by
$\sigma^\star(n)$, a functional of the magnitudes $\abs{\ak}$ alone, the branch by $C_g$, a
functional of the phases $\arg\rho_m$ alone (Prop.~\ref{prop:mechanisms}). A gap surviving in
cell~(I) is self-undermining: it places $p_g$ outside the smoothing class $\{p\star K\}$, but so
does the constant-size sampler that reproduces $p_g$. The split at $\sigma^\star = O(1)$ is
chosen for simplicity and is not the boundary of classical simulability. If the overlap
profile decays polynomially, the bound of Lemma~\ref{lem:compact}
gives $M^\star = \mathrm{poly}(n)$ whenever $\sigma^\star(n) = \mathrm{poly}(n)$,
so the cell~(I) construction extends verbatim to that larger class and cell~(II)
is correspondingly smaller. We state the bounded case, which is the one the examples
of Sec.~\ref{sec:examples} realize, and note that nothing in the argument turns
on the distinction.

For cell (II), in branch~(a) the certificate propagates to the gap, $\Phi_g\ge C_g$ by
Eq.~\eqref{eq:gap_lower}, and by Prop.~\ref{prop:mechanisms}~(ii) it requires
$\mathrm{Im}\wm\neq0$ at a regular lag, which the real window $W_m$ cannot supply. Branch~(b)
constrains the witness rather than the gap. $C_g$ is an $\ell^\infty$ quantity in the lag
whereas $\Phi_g$ is an $\ell^1$ one, so a residual of size $1/\sqrt{N}$ at each of $\Theta(N)$
lags contributes to $\Phi_g$ while $C_g\to0$, and the phase-free channels of
Prop.~\ref{prop:mechanisms}~(iii) remain open, detected without any phase hypothesis by
\begin{equation}
  \Phi_g\;\ge\;\tfrac12\,\max_m\big(\abs{\pg(m)}-\abs{\hat p(m)}\big)\,,
  \label{eq:real_channel_bound}
\end{equation}
valid since $\vert\hat K(m)\vert\le1$ for every probability kernel. A gap witnessed only this way
places $p_g$ outside the smoothing class but attests nothing a classical sampler could not do.
 Under the hypotheses of Prop.~\ref{prop:mechanisms}~(iii) both
channels close and the classification by $C_g$ is one by $\Phi_g$ at fixed $n$. We return to
both points in Sec.~\ref{sec:examples}.

\subsection{Transfer to the hypercube}
\label{sec:hypercube}
The classification in Sec.~\ref{sec:dichotomy} is aimed at filters operating on $\mathbb{Z}_N$.
However, the treatment uses the cyclic structure of the group only through character
orthogonality, the convolution theorem, Parseval's theorem, and Bochner's characterization of
realizable envelopes, all of which hold on any finite abelian group, so the framework transfers
with the dual group supplying the frequencies and the group difference replacing the lag $k-m$.
The one exception is Lemma~\ref{lem:compact}, whose mode count orders the lags along a line; on
the hypercube its role is played by a degree bound (Appendix~\ref{app:hypercube}). We sketch the
instance relevant to the didactic example that motivates this
work~\cite{belis2026spectralmethodscrucialmachine}, the hypercube $\mathbb{Z}_2^n$.

Consider a spectrum $\ak=2^{-n/2}\sum_x a_x(-1)^{k\cdot x}$ reached by a layer of Hadamard gates.
This is the Fourier transform on $\mathbb{Z}_2^n$. The equivalent of the 
frequency in the $\mathbb{Z}_N$ case is the Hamming weight $\abs{k}$, and the lag difference is the bitwise XOR
$k\oplus m$. We use a low-pass, $g_\eta(k)=(1-2\eta)^{\abs{k}}$ with flip rate
$0\le\eta\le 1/2$. In the taxonomy of Eq.~\eqref{eq:W_polar} this is a magnitude filter and its
pair-weight factorizes
\begin{equation}
  W_m(k)=\underbrace{(1-2\eta)^{\abs{m}}}_{=\,\hat K_\eta(m)}\;
  \underbrace{\big[(1-2\eta)^{2}\big]^{\abs{k\setminus m}}}_{\text{window}}\,,
  \label{eq:bitflip_fact}
\end{equation}
where $k$ and $m$ are identified with their supports in $[n]$, so $\abs{k\setminus m}$ counts
the ones of $k$ at positions where $m$ vanishes. The lag factor
is the transfer function of the i.i.d.\ bit-flip kernel
$K_\eta(z)=\eta^{\abs{z}}(1-\eta)^{n-\abs{z}}$~\cite{odonnell2014}.
That this kernel reproduces the filter applied to the
\emph{distribution} is already noted in
Ref.~\cite[Sec.~II\,E]{belis2026spectralmethodscrucialmachine}. The coherent route differs from
it only through the window inside Eq.~\eqref{eq:master_id}.

Two group properties then decide the classification. First, the phase channel is closed: the
characters are real, so $\hat p(m)$ and $\pg(m)$ are real for every input and every diagonal
filter, hence so is $\wm$, and $R_g(m)=0$ at every regular lag (Lemma~\ref{lem:real_wm} in
Appendix~\ref{app:hypercube}). The certificate is thus identically zero, here not only for
magnitude filters but for phase filters as well. As a result, the certified branch of cell~(II)
is empty, and any surviving gap runs through the phase-free channels of the uncertified branch
[cf.~Eq.~\eqref{eq:real_channel_bound}]. Second, matched cost is set by the weights
$A_w=\sum_{\abs{k}=w}\abs{\ak}^2$ through $\psucc=\sum_w(1-2\eta)^{2w}A_w$, so the affordable
flip rate is $\eta^\star(n)=\max\{\eta:\psucc(\eta)\ge\pstar\}$, the counterpart of the
affordable passband scale with small $\eta$ playing the role of large $\sigma$. A low-degree
input ($A_w=0$ for $w>w_0$ with $w_0=O(1)$) affords a constant flip rate and confines the output
to degree $\abs{m}\le2w_0$, which is a $\mathrm{poly}(n)$-coefficient object samplable in
$\mathrm{poly}(n)$ time. This is cell~(I) with ``$O(1)$ modes'' replaced by ``$O(1)$ degree'',
with the caveat that the confinement is inherited from the input, which already restricts the
unfiltered $\hat p$ to the same degrees, rather than imposed by the filter. A broadband input,
whose spectral mass escapes every fixed degree, instead forces the affordable flip rate to
vanish, $\eta^\star(n)\to0$, so no fixed degree is attenuated. This is the analogue of a
diverging $\sigma^\star(n)$.

On these two input classes, a magnitude filter on $\mathbb{Z}_2^n$ at matched cost is therefore
either a low-degree classical object or no smoothing at any fixed degree, with no phase branch
in either case, and whatever escapes the kernel surrogate is inherited from the input rather
than produced by the filter. Genuinely quantum spectral operations must consequently be phase
filters ($\abs{g}\equiv1$), the direction Ref.~\cite{belis2026spectralmethodscrucialmachine}
itself points to in suggesting a deterministic phase mask on the amplitudes. These incur no
post-selection cost, attenuate no degree here, and escape the compactness bound of
Lemma~\ref{lem:compact} on $\mathbb{Z}_N$. On the cube the certificate is blind to them as well,
so what such a filter buys there would have to be established by something other than $C_g$.

\section{Examples}
\label{sec:examples}

In this section, we probe the implications of the previous sections. In particular, we show that a spectrally incoherent input results
in a coherent residual that vanishes in expectation at every lag, while the gap itself survives as phase speckle. As a canonical example of a low-pass filter, we specify the properties of a Gaussian
magnitude filter and analyse a collection of quantum Born machines with random parameters and contrast them by trained counterparts.

\subsection{Spectrally incoherent input}
\label{sec:washout}
In the following,
we investigate the canonical example of an ensemble with i.i.d. random phases. In particular, consider a magnitude filter ($W_m$ real)
with input states $a_x=\sqrt{p(x)}\exp(i\varphi_x)$ with $\varphi_x$
independent and uniform. The spectral distribution of the input is flat in expectation, $\mathbb{E}[\abs{\hat{a}}^2]=1/N$
such that $\psucc\approx\sum_k\abs{g_\sigma(k)}^2/N$. The only way $\psucc$ can increase to reduce the cost is by widening the passband.
This qualifies the scenario as an example for the growing passband scale case.

As shown in Appendix~\ref{app:washout}, the coherent residual Eq.~\eqref{eq:coherent_residual} vanishes in expectation (w.r.t the phases)
and thus any deviation of the residual is a fluctuation around this vanishing mean. The fluctuations can be shown to be bounded by the
largest probability $p_{\rm max}=\text{max}_xp(x)$ such that at every regular lag $m$
\begin{equation}
  \mathbb{E}\big[\,\abs{R_g(m)}^{2}\,\big]\;\le\;\frac{p_{\max}}{\psucc^{2}}\,.
  \label{eq:washout_bound}
\end{equation}
The bound is therefore informative for input families whose magnitude profile delocalizes,
$p_{\rm max}\to0$, which we assume in the following. Note that this is a hypothesis on the
profile $p(x)$ and not a consequence of the flat expected spectrum
$\mathbb{E}[\abs{\hat a}^2]=1/N$, which the random phases produce for any $p$.
With $\psucc\ge \pstar$ fixed, the phase-charged residual then washes out asymptotically,
$R_g(m)\to 0$ in mean square at every regular lag. In particular, note that the vanishing
coherent residual is not caused by a vanishing
success probability since a smaller $\psucc$ would result in a larger $R_g$ [cf. Eq.~\eqref{eq:residual_final_form}].
The vanishing residual is thus caused by a vanishing $\mathrm{Im}\wm$ through phase concentration at fixed affordability. We
 stress that the same incoherent ensemble fed to a phase filter does not wash out: there $W_m$ is
complex, $R_g\not\to0$, and the certificate survives at every lag. From Eq.~\eqref{eq:washout_bound} we see that the certificate Eq.~\eqref{eq:gap_lower} vanishes and the family occupies the uncertified branch of
cell~(II). The bound Eq.~\eqref{eq:washout_bound} is per lag, however, and does not control
$\Phi_g$ (Appendix~\ref{app:washout}). 
\subsection{Gaussian low-pass}\label{sec:gaussian}
\begin{figure*}[t]
    \centering
    \includegraphics[width=0.82\textwidth]{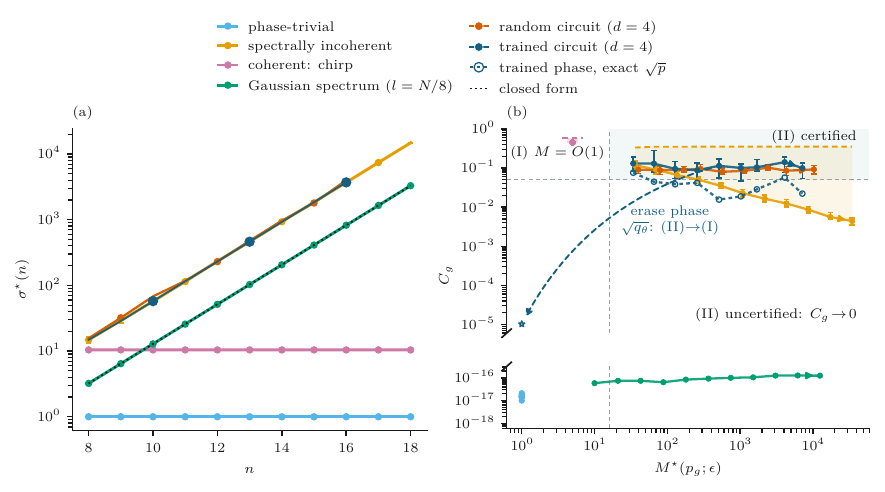}
  \caption{Input families for a Gaussian low-pass filter Eq.~\eqref{eq:gaussian_g} at
  the affordable passband scale $\sigma^\star(n)$ [Eq.~\eqref{eq:sigmastar}] at $\pstar=0.1$.
  (a)~The affordable passband scale for the constructed families and quantum circuits [cf. Sec.~\ref{sec:quantum_generative_models}]
  with random and with trained parameters. The closed form solution for Gaussian input spectra Eq.~\eqref{eq:gaussian_passband}
  is shown as dotted (black). (b) The phase certificate $C_g$
  [Eq.~\eqref{eq:gap_lower}] as a function of the modes $M^\star$
  [Eq.~\eqref{eq:modecount}] required to achieve an error of $\epsilon=0.05$.
  Each family is a track parametrised by $n$. The arrows point
  towards large $n$. For the coherent chirp and the spectrally incoherent input, the Gaussian-kernel upper
  bound on $\Phi_g$ is shown as a dashed line with the enclosed bracket indicated as shaded areas.
  The spectrally incoherent family and the circuits are shown as median over seeds (solid) and the error bars are the interquartile range.
  The visual classification into cells is
  schematic and not rigorous.
  The dashed arrow (dark blue) marks the phase erasure control experiment and the dotted track with the hollow markers is the complementary experiment
  described in the main text.}
  \label{fig:gaussian_low_pass}
\end{figure*}

To make the statements in Sec.~\ref{sec:results} concrete, we now fix a family of magnitude filters. As a prototypical example, we
consider a low-pass filter with a Gaussian profile
\begin{equation}
  g(k)=\exp(-k^2/2\sigma^2)\,,\label{eq:gaussian_g}
\end{equation}
such that
\begin{equation}
    W_m(k)=\exp(-\frac{m^2}{4\sigma^2})\exp(-\frac{(k-\frac{m}{2})^2}{\sigma^2})\,,\label{eq:gaussian_wk}
\end{equation}
with $g, W_m\in\mathbb{R}$. A straightforward calculation (Appendix~\ref{app:gauss}) yields the mode count Eq.~\eqref{eq:modecount}
\begin{equation}
    M^\star(\sigma;\epsilon)=O\!\big(\sigma\sqrt{\ln(1/\epsilon)}\big)\,.
    \label{eq:gaussian_mode_count}
\end{equation}
In particular, $M^\star$ depends on $n$ only through $\sigma$, consistent with the results of Sec.~\ref{sec:compact}.

For the Gaussian filter, the success probability Eq.~\eqref{eq:psucc} evaluates to
$\psucc(\sigma)=\sum_k\exp(-k^2/\sigma^2)\abs{\hat{a}(k)}^2$ increases for increasing $\sigma$.
The classification in terms of Theorem~\ref{thm:dichotomy} thus depends on the input spectrum.
We can make this concrete for a Gaussian input spectrum of width $l$ with magnitudes $\abs{\hat{a}(k)}^2= (1/l\sqrt{\pi})\exp(-k^2/l^2)$.
As shown in Appendix~\ref{app:gauss}, for sufficiently large $l$, we can turn the sum in $\psucc$ into an integral and
solve Eq.~\eqref{eq:sigmastar} for $\sigma^\star$, which is attained at $\psucc=\pstar$
\begin{equation}
    \sigma^\star = \frac{l \pstar}{\sqrt{1-p^{\star 2}}}\label{eq:gaussian_passband}
\end{equation}
From this, we can treat two separate cases. For broadband input, we have $l(n)\rightarrow\infty$
and thus $\sigma^\star\rightarrow\infty$. In this case the number of modes $M^\star$
is unbounded, placing the family in cell~(II) of Theorem~\ref{thm:dichotomy}, and the
branch is decided by the phases: a coherent input occupies the certified branch, a phase-trivial
input the uncertified branch with a vanishing certificate, and an incoherent input the uncertified branch
with its gap surviving as speckle (see previous section).
A narrow input spectrum, however, keeps $\sigma^\star$ and thus $M^\star$ bounded and the output
is a constant-size classical object according to cell~(I).

\subsection{Numerical Experiments}
\label{sec:numerics}

To gauge the validity of the analytical framework derived in Sec~\ref{sec:results}, we perform numerical 
experiments. 

\subsubsection{Controlled Input Families}
Figure~\ref{fig:gaussian_low_pass} (a) shows the affordable passband scale $\sigma^\star$ Eq.~\eqref{eq:sigmastar} of a Gaussian low-pass filter Eq.~\eqref{eq:gaussian_g} for a fixed
affordability threshold $\pstar = 0.1$ as a function
of the qubit number for different inputs. Three families share the same $n$-independent bimodal magnitude profile
$p(x)=\abs{a_x}^2$ as input. The profile is reflection-symmetric
about $x_0=0$ on the torus, $p(-x)=p(x)$. The families differ only in $\text{arg}\hat{a}(k)$: phase-trivial ($a_x=\sqrt{p(x)}$; light blue), spectrally incoherent (i.i.d. uniform
phases; yellow; median over $128$ seeds), and coherent chirp ($\arg a_x=\pi\gamma(x/N)^2$ with fixed $\gamma=48$; pink).  A fourth
family realizes the Gaussian input spectrum
$\abs{\hat a(k)}^2\propto e^{-k^2/l^2}$ with $l=N/8$ (green). 
It can be seen that the broad-band input of the Gaussian spectrum
and the spectrally incoherent inputs result in an unbounded $\sigma^\star\propto 2^n$, whereas the coherent
chirp and the phase-trivial input stay bounded. For the Gaussian spectrum the measured $\sigma^\star$ coincides with the closed form
Eq.~\eqref{eq:gaussian_passband} (dotted) to a relative deviation of
$10^{-8}$. 

Figure~\ref{fig:gaussian_low_pass} (b) shows the phase certificate $C_g$
Eq.~\eqref{eq:gap_lower} against the mode count $M^\star$ for each
family. The two axes are the two functionals Theorem~\ref{thm:dichotomy}. Here,
$M^\star$ reads the magnitudes through $\sigma^\star$, and $C_g$, a maximum over the regular
lags alone, reads the phases. To calculate
$M^\star$ we fix an approximation error of $\epsilon=0.05$. For
the spectrally incoherent family and the coherent chirp, we can in addition to the bounds Eq.~\eqref{eq:gap_lower} and Eq.~\eqref{eq:real_channel_bound},
obtain a Gaussian-kernel upper
bound by minimizing $\mathrm{TV}(p\star K_\sigma,p_g)$ over the
bandwidth $\sigma$ (dashed; for the chirp the minimizer is the degenerate $\sigma\to0$ kernel,
$p\star K=p$), which together with $C_g$ encloses the true gap $\Phi_g$.

Each family traces a short track as $n$ grows, and since cell membership is asymptotic it is
read off the direction of that track (arrows) rather than off any single point. The cell is
set by $\sigma^\star$. A track running right witnesses cell~(II),
but a bounded $M^\star$ is consistent with either cell, since Lemma~\ref{lem:compact} has no
converse. Within cell~(II) a height that holds marks the certified branch and one that falls
the uncertified one. Since $C_g\le\Phi_g$, a persistent height is also a persistent gap. 

Both the phase-trivial and the Gaussian-spectrum input have $R_g\equiv0$, so the shown
certificate is 0 at machine-precision. However, they land in different cells, since
membership is set by $\sigma^\star$. The phase-trivial input keeps $\sigma^\star$ bounded and
falls in cell~(I). Its empty certificate is not a vanishing gap as condition~(b) of
Prop.~\ref{prop:criterion} fails. To make this more concrete, we run a linear program over the full kernel class
$\mathcal{K}$ (Appendix~\ref{app:numerics}) which returns $\Phi_g=0.039$.
The Gaussian-spectrum input is
broadband, so $\sigma^\star\propto2^n$ places it in cell~(II), in the uncertified branch since
its real spectrum empties the certificate. There both conditions hold and the gap vanishes,
$\Phi_g=0$ up to the solver residual ($<10^{-9}$), so classical smoothing reproduces the output.

The incoherent track runs to the right and downwards. Its certificate decays
 as $2^{-0.49n}$ (fitted slope $-0.489\pm0.004$ over $8\le n\le18$, $128$ seeds), matching the
 $2^{-n/2}$ scale Eq.~\eqref{eq:washout_bound} sets, while the
 Gaussian-kernel upper bound for the
 same family plateaus. The best kernel recovers the
 envelope of $p_g$ and is left with its speckle. The family descends into the uncertified branch
 of cell~(II), its persistent gap carried by the phase-free channel. The chirp keeps $\sigma^\star$
bounded and is therefore in cell~(I), yet its gap persists, the
self-undermining case noted after Theorem~\ref{thm:dichotomy}.

\subsubsection{Quantum Born Machines}\label{sec:trained}
Having numerically tested the analytical regimes of Sec~\ref{sec:results}, we now analyze
distributions obtained from quantum circuits. To this end, we train a hardware-efficient
ansatz (depth-$d$ $R_yR_z$ rotation
layers with a CZ ring) to reproduce the fixed bimodal target $p(x)$ of
Fig.~\ref{fig:gaussian_low_pass} using the Kullback--Leibler divergence $\mathcal{L}$ of Sec.~\ref{sec:quantum_generative_models} as the loss function.
In particular, $\mathcal{L}$ only sees the distribution obtained through the Born rule $p_\theta(x)=\abs{a_x(\theta)}^2$
and is thus independent of any phases. We train circuits on 32 random parameter initializations with a depth of $d=4$. In Fig.~\ref{fig:gaussian_low_pass}, the untrained circuits are shown in orange
and the trained circuits in dark blue.

In Figure~\ref{fig:gaussian_low_pass} (a), both trained and untrained circuits have an
exponentially growing affordable passband scale $\sigma^\star$, which places both families in
cell~(II) of Theorem~\ref{thm:dichotomy}. The regular-lag certificate $C_g$ of the trained circuits is flat in $n$ (fitted slope $-0.04\pm0.02$). Hence $C_g\not\to0$ and the circuits occupy the \emph{certified}
branch of cell~(II). The mode count inherits the $\propto2^n$ growth
of the broadband circuit spectrum, so the track runs to the right, while its height stays at the
level of the same circuit with random parameters, a persistent certified gap at a growing
$M^\star$. Training barely displaces the track. The classification is not an artifact of an imperfect fit. The depth-4 models
reproduce the target to $\mathrm{TV}\approx0.02$, flat in $n$ (median over seeds between $0.020$
and $0.025$ for $8\le n\le16$). The mechanism is the one anticipated in
Prop.~\ref{prop:mechanisms}: the loss is a functional of only $p_\theta$, so it fixes the spectral
phase only through the component the profile itself dictates and neither needs nor removes the rest,
which is inherited at the initialization. Trained circuits are therefore
coherent inputs in the sense of Def.~\ref{def:incoherent}, and the certified branch of cell~(II)
is reached generically rather than by construction.
 
To isolate which half of the input carries the certified branch, we perform two control
experiments that cut a trained $n=12$ state along the magnitude/phase seam. First, erasing
the phases yields the twin $\sqrt{p_\theta}$, with an identical output distribution. It
collapses into cell~(I) at the narrowest admissible passband, $\sigma^\star=M^\star=1$
[dashed arrow in Fig.~\ref{fig:gaussian_low_pass}~(b)], its certificate falling to $10^{-5}$
(nonzero only through the residual fit error breaking the target's reflection symmetry) and
its exact gap to $\Phi_g\approx0.038$ (flat in $n$ and evaluated at $n\le10$, the sizes at which
the linear program for this control is affordable; Appendix~\ref{app:numerics}), the phase-free
floor of the profile and an order of
magnitude below the $\approx0.31$ of the state it was built from. Second, the converse
surgery [dotted track, hollow markers]: keeping the trained phases but substituting the exact
target magnitudes, $a_x=\sqrt{p(x)}\exp(i\arg a_x(\theta))$, removes the residual misfit. The
remaining profile is reflection-symmetric and hence phase-trivial on its own. The gap survives
intact, $\Phi_g\approx0.31$, within $0.1\%$ of the trained value in the median, at a passband
scale and mode count indistinguishable from the trained state, so the track stays in
cell~(II). The two controls bracket the classification from either side. It is only decided 
by the phases.

This finding does not promote the magnitude filter into a quantum resource. By Eq.~\eqref{eq:gap}, $\Phi_g$ is the distance of the
filter's output from its purpose of smoothing, so for a magnitude filter \emph{working as smoothing} and
\emph{having a gap} are mutually exclusive. On the trained models the filter does not smooth. The
filtered output lands at $\mathrm{TV}\approx0.31$ from the target the model was trained to
represent. Moreover, it lands fifteen times further from that target than the unfiltered model
does, at ten times the sampling cost, and the gap scatters
across identically trained seeds ($0.14$--$0.55$ at $n=12$), reflecting the phase configuration
the initialization happened to supply rather than any structure related to the task. Both readings of what the filter is for lead to the same place. If the goal is smoothing, the
filter delivers it only in the regime $\Phi_g\to0$ which is the regime in which a
classical kernel applied to the unfiltered samples delivers it as well, and does so without
paying the post-selection cost. If the goal is instead the filtered distribution $p_g$ in its
own right, then claiming an advantage means claiming that $p_g$ is hard to sample classically.
That is a statement about the input state family, which can be prepared and sampled without ever
applying the filter. The magnitude filter adds only the real window $W_m$ and the cost $\psucc$.
Either way, the quantum candidate is the input state, not the filter.

\section{Discussion and Conclusion}
\label{sec:discussion}
In this work we studied spectral filtering of generative quantum machine learning models. 
Spectral filtering is proposed as a regularizer, its job is
to smooth the output distribution of a generative model without destroying whatever makes that model
worth running on a quantum device. We have investigated whether it can do both, and whether the coherent
operation produces anything a classical competitor with the same interface cannot. We made the
comparison explicit for quantum circuit Born machines, whose output distributions are blind to the
phases of the amplitudes that generate them, measuring coherent filtering against convolution of the
unfiltered samples with a symmetric probability kernel at matched sampling cost. For magnitude
filters, the coherent counterpart of classical smoothing, the two requirements are in tension.
At a fixed affordability threshold, either the
filtered output is already a classical object, reproducible by smoothing the unfiltered samples or
specified by a constant number of Fourier modes and efficiently samplable from them, or the filter provides
no attenuation at any fixed frequency: applied to the broadband spectra that trained circuits actually produce, it is
affordable only with a passband so wide that it attenuates nothing at any fixed frequency.
Regularization, affordability, and any separation from classical simulation cannot be achieved together.
Where a separation does survive the operation, it is not manufactured by the filter. The residual
that carries it is sourced entirely by the spectral phase of the input state, and replacing that
phase by a linear one while leaving the output distribution untouched, which for the
reflection-symmetric target studied here is simply erasing the computational-basis phases, collapses
a trained model into the classically reproducible cell. Numerical experiments on controlled input families and on trained and untrained
Born machines illustrate both cells, and on the trained circuits the point is visible without any
simulability argument at all, since the filtered model sits fifteen times further from its own
training target than the unfiltered one at ten times the sampling cost. 

Several limitations qualify these conclusions. We work out the cyclic group $\mathbb{Z}_N$ in
detail and only sketch the transfer to the hypercube $\mathbb{Z}_2^n$, even though every ingredient
we use is available on any finite abelian group. The analysis is restricted to filters diagonal in
the Fourier basis, so spectral operations that couple modes are outside its scope. The classical
competitor is convolution with a symmetric probability kernel, which is the operation the filter
advertises, rather than arbitrary classical post-processing of the samples. Dropping reflection
symmetry costs the competitor almost nothing on every instance we solve
(Appendix~\ref{app:numerics}), but a gap measured against $\mathcal{K}$ is a separation from
smoothing and not, on its own, from classical simulation. Finally, the gap Eq.~\eqref{eq:gap} is a
statement about distributions in the large-sampling limit. At a finite shot budget $\Phi_g$ must be
resolved against shot noise, at a rate itself reduced by the post-selection cost $\psucc$, and how
the classification degrades there we leave open.

What the dichotomy delineates is what is \emph{not} a quantum resource, which is also where it
points. Within the diagonal family, phase filters with unit modulus are the only spectral
operations left. They incur no post-selection cost, need not admit a kernel surrogate, and escape the
compactness bound of Lemma~\ref{lem:compact}. Whether they can be turned into a design principle
with a useful inductive bias, rather than merely an unsimulable one, is the natural next question.
A second direction follows from the mechanism of Sec.~\ref{sec:trained}: the spectral phases that
decide the classification are, on the circuits studied here, invisible to a Born-rule loss and are
inherited from the initialization. This is an empirical statement about trained circuits rather than
a structural one, since the component of $\arg\ak$ sourced by the asymmetry of the profile is a
functional of $p_\theta$ and hence is seen by such a loss; our controls show that component to be
negligible here, $10^{-5}$ against a gap of $0.31$.
Training objectives that see the phase, rather than only the output distribution,
would make that degree of freedom a design variable instead of an accident.

\begin{acknowledgements}
  This work has been supported by the Dieter Schwarz Foundation through
  the Fraunhofer Heilbronn Research and Innovation Center Applied Quantum AI.
   Large language models were used in the preparation of this
  manuscript for literature research, editorial work on the text, code prototyping,
  and as a verification device on the technical arguments. All results were derived
  and verified by the author, who takes full responsibility for the content.
\end{acknowledgements}

\appendix
\onecolumngrid

\section{Framework}
\label{app:framework}
This appendix supports the framework described in Sec.~\ref{sec:framework}

\subsection{Classical Competitor}
We define the classical competitor as smoothing operations induced by a symmetric probability kernel as follows.

\begin{definition}[Classical post-processing kernels]\label{def:kernels}
A \emph{probability kernel} on $\mathbb{Z}_N$ is a map $K:\mathbb{Z}_N\to\mathbb{R}$ with
$K(x)\ge0$ and $\sum_xK(x)=1$, acting on the model as $G_K:p\mapsto p\star K$. Operationally $G_K$ draws
$x\sim p$ and $u\sim K$ independently and returns $x+u\bmod N$, so it is classical post-processing of
the samples for \emph{every} such $K$. We write
\begin{equation}
  \begin{aligned}
    \mathcal{K}_+&\equiv\Big\{K\ \Big|\ K\ge0,\ \sum_xK(x)=1\Big\}\,,\\
    \mathcal{K}&\equiv\big\{K\in\mathcal{K}_+\ \big|\ K(-x)=K(x)\big\}\,,
  \end{aligned}
  \label{eq:kernel_classes}
\end{equation}
for the translation-equivariant post-processings and for their reflection-symmetric subclass,
the \emph{smoothing kernels}. Within $\mathcal{K}_+$ the transfer function $\hat K$ is real and 
even within $\mathcal{K}$.
\end{definition}
Reflection symmetry is what separates smoothing from shifting: the extreme points of
$\mathcal{K}_+$ are the deterministic translations $x\mapsto x+x_0$, and $\mathcal{K}$ retains only their
symmetric mixtures. We take $\mathcal{K}$ as the classical baseline throughout, since it is the class to
which the filter's advertised semantics Eq.~\eqref{eq:convolution_filter} belongs. Since $\mathcal{K}\subset\mathcal{K}_+$ gives
$\Phi_g\ge\Phi_g^+\equiv\inf_{K\in\mathcal{K}_+}\mathrm{TV}(p\star K,p_g)$, every classicality
conclusion drawn against $\mathcal{K}$ holds a fortiori against $\mathcal{K}_+$, whereas a positive
gap must be checked against the larger class. On all instances reported here the median
$\Phi_g^+/\Phi_g$ exceeds $0.93$ for every input family with a nonzero gap, and no instance
changes its cell assignment (Appendix~\ref{app:numerics}). This justifies the restriction to $\mathbb{K}$ in the main text.

\subsection{Derivation of Eq.~\eqref{eq:master_id}}
We use the inverse of Eq.~\eqref{eq:fourier_transf} to write the filtered amplitudes $a_g(x)=\tfrac{1}{\sqrt N}\sum_k \hat a_g(k)\,\omega^{kx}$.
The filtered modes then have the form
\begin{equation}
  \hat p_g(m)=\sum_x \abs{a_g(x)}^2\omega^{-mx}
  =\frac1N\sum_{k,k'}\hat a_g(k)\hat a^*_g(k')\sum_x \omega^{(k-k'-m)x}
  =\sum_k \hat a_g(k)\hat a^*_g(k-m)\,,\label{eq:master_id_derivation}
\end{equation}
where we have used $\sum_x\omega^{(k-k'-m)x}=N\,[\,k-k'-m\equiv 0\,]$ in the last equality. 

For a diagonal spectral filter as introduced in Sec.~\ref{sec:spectral_filtering}, normalized to
$\norm{g}_\infty=1$, the map $\hat a\mapsto g\hat a$ can be realized on a single ancilla. We write $D_g=\sum_k g(k)\ket k\bra k$ for the corresponding operator, diagonal in the Fourier basis. Since
$\abs{g(k)}\le1$, the filter admits a dilation
$V\ket{0}_a\ket{k}=g(k)\ket{0}_a\ket{k}+(1-\abs{g(k)}^2)^{1/2}\,\ket{1}_a\ket{k}$
which block-encodes the filter
\begin{equation}
  \big(\bra{0}_a\otimes I\big)\,V\,\big(\ket{0}_a\otimes I\big)=D_g\,,
  \label{eq:heralding}
\end{equation}
so the outcome $\ket{0}_a$ occurs with probability $\norm{g\hat a}^2=\psucc$
[Eq.~\eqref{eq:psucc}] and leaves the normalized amplitudes $\hat a_g(k)=g(k)\ak/\sqrt{\psucc}$.
This gives $\psucc$ its two readings, the retained spectral mass and the heralding probability,
with $\psucc=1$ exactly when $\abs{g(k)}=1$ on the support of $\hat a$. Repeating until success
costs $1/\psucc$ state preparations, which amplitude amplification~\cite{lomonaco2002quantum} reduces to
$O(1/\sqrt{\psucc})$. No result depends on which, since only $\psucc\ge\pstar$ at $n$-independent
$\pstar$ is used. Substituting $\hat a_g$ into Eq.~\eqref{eq:master_id_derivation} results in Eq.~\eqref{eq:master_id} with the definitions of $W_m(k)$ and $\rho_m(k)$ as given in the main text. 

\section{Simulability criterion}\label{app:criterion}
In this appendix, we derive the simulability criterion of Prop.~\ref{prop:criterion}. The coherent residual Eq.~\eqref{eq:coherent_residual}
vanishes if the classical filter matches the quantum filter.
Expanding $h^\star$ and identifying the zero point of its derivative gives
\begin{align}
&\partial_h \left(\abs{\hat{p}_g(m)}^2-2h\,\text{Re}(\hat{p}_g(m)\hat{p}^*(m))+h^2\abs{\hat{p}(m)}^2\right)=0\\
&\iff h^\star(m) = \text{Re}\left(\frac{\hat{p}_g(m)}{\hat{p}(m)}\right)\,.\label{eq:hstar}
\end{align}
Substituting this into Eq.~\eqref{eq:coherent_residual} results in Eq.~\eqref{eq:residual_final_form}.
Using this equation, we study the gap as defined in Eq.~\eqref{eq:gap}, which can be formulated
as $\Phi_g=\text{inf}_K\norm{p_g-p\star K}_1/2$. The kernel set $\mathcal{K}$ is a closed
and bounded, hence compact, subset of $\mathbb{R}^N$, and $K\mapsto\text{TV}(p\star K,p_g)$ is
continuous, so the infimum is attained and $\Phi_g=0$ holds exactly when some $K\in\mathcal{K}$
reproduces $p_g$. On the finite group $\mathbb{Z}_N$, the two norms are equivalent,
$\norm{f}_2\leq\norm{f}_1\leq\sqrt{N}\norm{f}_2$, so they vanish together and $\norm{p_g-p\star
K}_2^2$ characterizes $\Phi_g=0$ in both directions. The equivalence constant is $\sqrt N$; hence
the $L^2$ route says nothing quantitative or asymptotic about $\Phi_g$ itself (see the scope
paragraph of App.~\ref{app:washout}). Working in $L^2$
allows us to use Parseval's theorem with the convention
\begin{equation}
  \sum_x \abs{f(x)}^2=\frac{1}{N}\sum_m \abs{\hat{f}(m)}^2\,.
  \label{eq:parseval}
\end{equation}
This yields
\begin{equation}
  \norm{p_g-p\star K}_2^2 = \frac{1}{N}\sum_m\abs{\pg(m)-\widehat{p\star K}(m)}^2 = \frac{1}{N}\sum_m\abs{\pg(m)-\hat{K}(m)\hat{p}(m)}^2\,,
\end{equation}
where $\hat{K}(m)\in\mathbb{R}$ is real and even for symmetric kernels $K$, and varies with the lag
$m$; the scalar $h$ of Eq.~\eqref{eq:hstar} is the minimization variable at a fixed lag.
Rearranging Eq.~\eqref{eq:coherent_residual}
and subtracting $\hat{K}(m)\hat{p}(m)$ from both sides gives
\begin{equation}
  \pg(m) -\hat{K}(m)\hat{p}(m) = R_g(m) + \big(h^\star(m) - \hat{K}(m)\big)\hat{p}(m)\,,\label{eq:res_decomposition}
\end{equation}
such that
\begin{equation}
  \abs{\pg(m)-\hat{K}(m)\hat{p}(m)}^2 = \abs{R_g(m) + \big(h^\star(m) - \hat{K}(m)\big)\hat{p}(m)}^2 = \abs{R_g(m)}^2 + \abs{h^\star(m) - \hat{K}(m)}^2\abs{\hat{p}(m)}^2\,.\label{eq:quadratic_bound_on_rg}
\end{equation}
Here, the cross terms in the second equality vanish since the first term is an imaginary multiple of $\hat{p}(m)$ by Eq.~\eqref{eq:residual_final_form}, and the second term is a real multiple of $\hat{p}(m)$. We thus have
\begin{equation}
  \norm{p_g-p\star K}_2^2 = \frac{1}{N}\sum_m\abs{R_g(m)}^2 + \frac{1}{N}\sum_m\abs{h^\star(m) - \hat{K}(m)}^2\abs{\hat{p}(m)}^2\label{eq:two_term_vanish_classicality}\,.
\end{equation}
At a degenerate lag, the second term is absent, and the first contributes $\abs{\pg(m)}^2$ for every $K$. Hence, $\Phi_g=0$ requires $\pg(m)=0$
at every degenerate lag, which is the degenerate half of condition (a) of Prop.~\ref{prop:criterion}.
For regular lags $m\in\mathcal{S}$ [Eq.~\eqref{eq:regular_lags}], a vanishing gap
requires both terms to vanish. Since $\hat{p}(m)\neq0$ and $\psucc>0$, the first term
vanishes exactly when $\wm\in\mathbb{R}$ by Eq.~\eqref{eq:residual_final_form}, which is the regular
half of condition~(a) in both directions. The second term in Eq.~\eqref{eq:two_term_vanish_classicality} vanishes if $\hat{K}(m) = h^\star(m)$ at every $m\in\mathcal{S}$, i.e., if $h^\star$ can
be realized by a symmetric kernel transformation, which is condition (b) in Prop.~\ref{prop:criterion}. By Bochner's theorem on $\mathbb{Z}_N$
~\cite{rudin1960}, the realizable envelopes $\hat{K}$ of symmetric probability kernels are the symmetric positive-definite sequences normalized to $\hat{K}(0)=1$. A vanishing second term therefore forces $h^\star$
 to coincide with one such envelope on the regular lags, which is condition~(b).

Three properties now immediately follow. Since $\hat p(0)=\pg(0)=1$, Eq.~\eqref{eq:hstar} gives
$h^\star(0)=1$, so the Bochner normalization is not an extra assumption. Since $p$ and $p_g$ are
real, $\hat p(-m)=\hat p(m)^*$ and $\pg(-m)=\pg(m)^*$, so $h^\star$ is automatically real and even,
and the restriction to symmetric kernels in condition~(b) costs nothing. Finally, every probability
kernel obeys $\abs{\hat K(m)}\leq\sum_xK(x)=1$, so $\abs{h^\star(m)}\leq1$ is necessary, which is
the amplification channel of Sec.~\ref{sec:dichotomy}, appearing already here.

Collecting the two directions: if $\Phi_g=0$, both terms of
Eq.~\eqref{eq:two_term_vanish_classicality} vanish at the attaining $K$, which is condition~(a) at
every lag and condition~(b) on $\mathcal{S}$. Conversely, (a) and (b) provide a $K^\star\in\mathcal{K}$
with $\hat K^\star=h^\star$ on $\mathcal{S}$; both terms vanish, and
$\norm{p_g-p\star K^\star}_2=0$, hence $\Phi_g=0$.\qed

\section{Spectral Compactness}
\label{app:compact}

\emph{Proof of Lemma~\ref{lem:compact}.} The modulus of $W_m$ bounds the
Fourier support of $\hat{p}_g$, so the magnitude profile that sets how many modes are needed to describe
 the filtered distribution. 
For any filter $g$, we bound Eq.~\eqref{eq:master_id}
\begin{equation}\label{eq:perm}
  \abs{\pg}\;\le\tfrac1\psucc\sup_k|W_m(k)|\sum_k|\rho_m(k)|\le\;\frac{1}{\psucc}\,\sup_k|W_m(k)|
  \;=\;\frac{1}{\psucc}\,\sup_k\big(|g(k)|\,|g(k-m)|\big)\,.
\end{equation}
In the second inequality, we have used $\sum_k\abs{\rho_m(k)}=\sum_k|\hat a(k)||\hat a(k-m)|\le\|\hat a\|^2=1$ by Cauchy-Schwarz. Hence, keeping all frequencies with $|m|\le M$,
\begin{equation}\label{eq:tvtail}
  \text{TV}\!\big(p_g,\ p_g^{\le M}\big)=\tfrac12\|p_g-p_g^{\le M}\|_1\le\;\tfrac12\Big(\textstyle\sum_{|m|>M}|\hat p_g(m)|^2\Big)^{1/2}
  \;\le\;\frac{1}{2\,\psucc}\Big(\textstyle\sum_{|m|>M}\sup_k|W_m(k)|^2\Big)^{1/2}.
\end{equation}
Here we have used Parseval's theorem Eq.~\eqref{eq:parseval} to bound the total variation by the sum of absolutes of the Fourier transformed probabilities $\|f\|_1\le\sqrt{N}\norm{f}_2=(\sum_m|\hat f(m)|^2)^{1/2}$ for the first inequality and the second inequality follows from Eq.~\eqref{eq:perm}.

Substituting the hypothesis $\sup_k\big(\abs{g_\sigma(k)}\,\abs{g_\sigma(k-m)}\big)\le
s(\abs{m}/\sigma)$ from Lemma~\ref{lem:compact} into Eq.~\eqref{eq:tvtail}, we use that $s$ is non-increasing and replace the sum with an integral, we obtain
\begin{equation}
  \text{TV}(p_g,p_g^{\leq M})\leq\frac{\sqrt{\sigma\,T(M/\sigma)}}{\sqrt2\,\psucc}\,.
  \label{eq:T_bound}
\end{equation}
Here we have defined $T(a):=\int_a^\infty s(u)^2\,du$. We are interested in the number of modes $M^\star$ Eq.~\eqref{eq:modecount} required to approximate the distribution
$p_g$ by the truncated distribution $p_g^{\leq M}$ Eq.~\eqref{eq:truncation} to error $\epsilon$, i.e., $\text{TV}(p_g,p_g^{\leq M})\leq \epsilon$.
Substituting Eq.~\eqref{eq:T_bound} then yields
\begin{equation}
  M^*(\sigma;\epsilon) \leq \sigma\,T^{-1}\!\Big(\frac{2\,(\pstar\epsilon)^2}{\sigma}\Big)\,,
  \label{eq:m_star}
\end{equation}
where we have used $\psucc \geq \pstar$ such that the bound depends on the state only through $\psucc\ge \pstar$.\qed

The decay of $\hat p_g$ is controlled purely by the filter's overlap width ($\approx 2\sigma$ for families $g_\sigma$ as defined in the main text).

\section{Proofs of Prop.~\ref{prop:mechanisms} and Theorem~\ref{thm:dichotomy}}
\label{app:dichotomy}
Throughout this appendix $g$ is a magnitude filter, so that
$W_m(k)=\abs{g(k)}\abs{g(k-m)}\in\mathbb{R}_{\ge0}$, and $\{g_\sigma\}$ is a family with passband
scale $\sigma$ in the sense of Sec.~\ref{sec:compact}. Furthermore, we assume the \emph{stretched form}
$\abs{g_\sigma(k)}=G(\abs{k}/\sigma)$ for a fixed profile $G:[0,\infty)\to[0,1]$ that is
continuous, non-increasing, and normalized to $G(0)=1$; in particular, $\abs{g_\sigma(k)}$ is of
order unity for $\abs{k}\lesssim\sigma$ and decays for $\abs{k}\gtrsim\sigma$. Then,
$\psucc(\sigma)$ is non-decreasing in $\sigma$, so Eq.~\eqref{eq:sigmastar} is well posed.

\paragraph{Derivation of the phase certificate Eq.~\eqref{eq:gap_lower}.}
Fix $K\in\mathcal K$ and set $f=p_g-p\star K$. Linearity and the convolution theorem give
$\hat f(m)=\pg(m)-\hat K(m)\hat p(m)$, and $\abs{\omega^{-mx}}=1$ gives
$\abs{\hat f(m)}\le\norm{f}_1=2\,\mathrm{TV}(p_g,p\star K)$. At a regular lag the orthogonal
decomposition Eq.~\eqref{eq:quadratic_bound_on_rg} of App.~\ref{app:criterion} bounds the same
quantity from below,
\begin{equation*}
  \abs{\hat f(m)}^2=\abs{R_g(m)}^2+\abs{h^\star(m)-\hat K(m)}^2\abs{\hat p(m)}^2
  \;\ge\;\abs{R_g(m)}^2\,,
\end{equation*}
so $\mathrm{TV}(p_g,p\star K)\ge\tfrac12\abs{R_g(m)}$ for every $m\in\mathcal S$ and every
$K\in\mathcal K$. Maximizing over $m$ and taking the infimum over $\mathcal K$ yields
$\Phi_g\ge C_g$.

\begin{proof}[Proof of Prop.~\ref{prop:mechanisms}]
\emph{(i)} The success probability
Eq.~\eqref{eq:psucc} sees the input only through the spectral
magnitudes $\abs{\ak}$, hence so does $\sigma^\star(n)$ defined by Eq.~\eqref{eq:sigmastar}.
Evaluating the bound Eq.~\eqref{eq:m_star} at $\sigma=\sigma^\star(n)$, where $\psucc\ge\pstar$
holds by construction, gives $M^\star(p_g;\epsilon)\;\le\;\sigma^\star T^{-1}(2(\pstar\epsilon)^2/\sigma^\star)$
in which the filter appears only through the overlap profile $s$ of its modulus, the state only
through $\psucc\ge\pstar$, and the grid size $N=2^n$ not at all. Since $T$ is non-increasing, so
is $T^{-1}$, and the right-hand side is non-decreasing in $\sigma^\star$. At fixed $\epsilon$ and
$\pstar$ a bound $\sigma^\star(n)\le\sigma_0$ therefore gives
$M^\star\le\sigma_0T^{-1}\big(2(\pstar\epsilon)^2/\sigma_0\big)$ uniformly in $n$ and for
arbitrary input phases.

\emph{(ii)} At a regular lag $\hat p(m)\neq0$ and
$\psucc>0$, so $R_g(m)\neq0\iff\mathrm{Im}\wm\neq0$ by Eq.~\eqref{eq:residual_final_form}. For a magnitude filter $\theta\equiv0$ in
Eq.~\eqref{eq:W_polar} and $W_m(k)\in\mathbb{R}$, so that with
$\wm=\sum_kW_m(k)\rho_m(k)/\hat p(m)$,
\begin{equation}
  \mathrm{Im}\wm=\frac{1}{\abs{\hat p(m)}^2}\sum_k W_m(k)\,
  \mathrm{Im}\big[\rho_m(k)\,\hat p^*(m)\big]\,.
  \label{eq:im_source}
\end{equation}
The real weights rotate no term, so the imaginary part is sourced exclusively by the arguments of
$\rho_m(k)=\ak\hat a^*(k-m)$ relative to $\hat p(m)=\sum_k\rho_m(k)$, i.e.\ by the spectral phase
of the input. 

On a phase-trivial family, inserting $\ak=e^{i\alpha}\omega^{tk}r_k$ makes the collinearity explicit,
\begin{equation}
  \rho_m(k)=\ak\hat a^*(k-m)=e^{i\lambda m}\,r_kr_{k-m}\,,\qquad \lambda\equiv\frac{2\pi t}{N}\,,
  \label{eq:rho_collinear}
\end{equation}
so the global phase $\alpha$ cancels and every term carries the same $k$-independent phase
$e^{i\lambda m}$, which is exactly the hypothesis at work. Consequently
$\hat p(m)=e^{i\lambda m}\sum_kr_kr_{k-m}$ and the terms $\rho_m(k)$ are collinear with their own
sum, $\rho_m(k)\,\hat p^*(m)=r_kr_{k-m}\sum_{k'}r_{k'}r_{k'-m}\in\mathbb{R}$ for all $k$
at every regular lag, so every summand of
Eq.~\eqref{eq:im_source} vanishes, $\wm\in\mathbb{R}$, and hence $R_g\equiv0$ and $C_g=0$
for every $n$.

\emph{(iii)} The lower bound is the certificate bound above. For the upper bound,
condition~(b) provides $K^\star\in\mathcal K$ with $\hat K^\star=h^\star$ on $\mathcal S$. In the
exact decomposition Eq.~\eqref{eq:two_term_vanish_classicality} the envelope term then vanishes
on $\mathcal S$, and the degenerate lags contribute $\abs{\pg(m)}^2=0$ by hypothesis, so
$\norm{p_g-p\star K^\star}_2^2=\tfrac1N\sum_{m\in\mathcal S}\abs{R_g(m)}^2$. With
$\norm{f}_1\le\sqrt N\norm{f}_2$ and $\Phi_g\le\mathrm{TV}(p\star K^\star,p_g)$ this is the
right-hand side of Eq.~\eqref{eq:gap_sandwich}. Finally, $C_g=0$ forces $R_g(m)=0$ at every
$m\in\mathcal S$ and hence $\Phi_g=0$ by that upper bound, while $\Phi_g=0$ forces $C_g=0$ by the
lower bound.
\end{proof}

\begin{proof}[Proof of Theorem~\ref{thm:dichotomy}]
\emph{Exclusive and exhaustive.} Split on whether $\sigma^\star(n)$ is bounded, giving
cell~(I) or cell~(II); within cell~(II), split on whether $C_g\to0$, giving the certified
branch~(a) or the uncertified branch~(b).

\emph{Cell~(I).} Part~(i) of Prop~\ref{prop:mechanisms} at $\sigma^\star(n)=O(1)$ gives $M^\star(p_g;\epsilon)=O(1)$
uniformly in $n$ and for arbitrary input phases. The truncation $p_g^{\le M^\star}$ is real and normalized but need not be pointwise nonnegative.

Samplability: We expand $\pg(m)$ into real and imaginary parts, so the truncation Eq.~\eqref{eq:truncation} becomes

\begin{equation}
  p_g^{\le M}(x)=\frac1N\Big(1+2\sum_{m=1}^{M}\big[\mathrm{Re}\hat p_g(m)\cos\tfrac{2\pi mx}{N}-\mathrm{Im}\hat p_g(m)\sin\tfrac{2\pi mx}{N}\big]\Big)\,.
\end{equation}
Here we have used that $\pg(-m)=\pg(m)^*$ for real $\pg$ and the truncation window $\abs{m}<M$ is symmetric in $m$.
From this, it follows that $\pg^{\le M}$ is a real trigonometric polynomial, normalized by its
$m=0$ mode. Note that it can become negative below zero wherever $p_g$ is small, so it is only an approximation of distribution, not a
distribution. This can be resolved by clipping the negative part and renormalizing,
\begin{equation*}
  q=\frac{\max\big(p_g^{\le M^\star},0\big)}{Z}\,,\qquad
  Z=\sum_x\max\big(p_g^{\le M^\star}(x),0\big)\,,
\end{equation*}
Clipping sends each negative point to $0$, which
is nearer the nonnegative $p_g$ than where it was before, so it cannot increase the error,
$\norm{\max(p_g^{\le M^\star},0)-p_g}_1\le\norm{p_g^{\le M^\star}-p_g}_1\le2\epsilon$.
Since $p_g^{\le M^\star}$ already sums
to one, $Z-1$ is exactly the negative mass the clip removes, bounded by the same
discrepancy, $Z-1\le2\epsilon$. Renormalizing displaces the vector by
$\norm{q-\tilde p}_1=\abs{1-Z}$, so $\norm{q-p_g}_1\le4\epsilon$ and
$\mathrm{TV}(q,p_g)\le2\epsilon$.

The distribution $q$ is efficiently
samplable: partial sums of the $O(1)$ retained modes are geometric series in closed form, a real
trigonometric polynomial of degree $M^\star$ changes sign at most $2M^\star$ times, so the
clipped region is a union of $O(1)$ arcs computable to precision $\delta$ in
$\mathrm{poly}(n,\log(1/\delta))$ time, and inverse-transform sampling from $q$ follows.
The conclusion uses no property of $C_g$ or $\Phi_g$.

\emph{Cell~(II)} If $\sigma^\star(n)$ is unbounded then, by the stretching form of
the family, $\abs{g_{\sigma^\star}(k)}\to 1$ at every fixed $k$
along the subsequence realizing the divergence, hence
$W_m(k)\to1$ on every fixed set of frequencies and lags.

\emph{Branch~(a), certified.} By part~(ii), $C_g>0$ requires $\mathrm{Im}\wm\neq0$ at some
regular lag. The window $W_m$ is real and by Eq.~\eqref{eq:im_source} cannot supply it, so it is
carried by $\arg\rho_m$, that is by the spectral phase of the input. By Eq.~\eqref{eq:gap_lower}, $\Phi_g\ge C_g$, so $C_g\not\to0$ forces
$\Phi_g\not\to0$: the output stays bounded away from every classical smoothing, and the
separation is carried by the spectral phase of the input.

\emph{Branch~(b), uncertified.} $C_g\to0$ means $\max_{m\in\mathcal S}\abs{R_g(m)}\to0$: no
regular lag separates $p_g$ from the smoothing class. If $\Phi_g\to0$, then by
Eq.~\eqref{eq:gap} and the attainment of the infimum (App.~\ref{app:criterion}) there are
symmetric probability kernels $K_n$ with $\mathrm{TV}(p\star K_n,p_g)\to0$. The unfiltered samples reproduce $p_g$ in the limit of large $n$
at no post-selection cost. This is consistent with the absence of attenuation since a filter that tends to the identity on every fixed band
is matched by the trivial kernel $K=\delta_0$, and a kernel that smooths at a fixed
frequency would contradict Eq.~\eqref{eq:no_attenuation}.

If instead $\Phi_g\not\to0$, then along any subsequence with $\Phi_g$ bounded away from zero
either a hypothesis of part~(iii) of Prop.~\ref{prop:mechanisms} fails, i.e.\ $\pg(m)\neq0$ at a
degenerate lag or $h^\star$ is not a realizable envelope, or both hypotheses hold and
Eq.~\eqref{eq:gap_sandwich} forces $\sum_{m\in\mathcal S}\abs{R_g(m)}^2\ge4\Phi_g^2\not\to0$,
which is compatible with $C_g\to0$ since residuals of size $N^{-1/2}$ at the
$\abs{\mathcal S}=\Theta(N)$ regular lags suffice. The two channels need no
$\mathrm{Im}\wm\neq0$ to open, and the third mechanism requires it at no lag above the vanishing
certificate, so none certifies coherence through a single-lag observable. What all three certify
is that $p_g$ has left the smoothing class, i.e.\ that the filter is not doing what it advertises.

In neither cell is the separation attributable to the filter. The mode count of part~(i) depends
on $g$ only through $\abs{g}$ and on the state only through $\psucc\ge\pstar$, while the residual
that carries a certified gap is sourced entirely by the phases of $\rho_m$. The filter supplies
the real window $W_m$ and the post-selection cost $\psucc$, and nothing else.
\end{proof}

\section{Spectral filtering on the hypercube}
\label{app:hypercube}
Character orthogonality, the convolution theorem, Parseval's theorem and Bochner's
characterization exist on any finite abelian group~\cite{rudin1960}, so
Eqs.~\eqref{eq:master_id}--\eqref{eq:coherent_filter}, Prop.~\ref{prop:criterion} and the
washout bound Eq.~\eqref{eq:washout_bound} hold there verbatim, with the lag difference $k-m$
read as the group difference, on $\mathbb{Z}_2^n$ the bitwise XOR $k\oplus m$. The exception is
Lemma~\ref{lem:compact}, which counts lags along a line: on the Boolean cube the lags at
distance $w$ number $\binom{n}{w}$, so the tail sum of Eq.~\eqref{eq:tvtail} is no longer
controlled by the decay of $\sup_k\abs{W_m(k)}$ alone. Its role is played by the degree bound
below, which, unlike Lemma~\ref{lem:compact}, is a hypothesis on the input rather than on the
filter. The characters $\chi_k(x)=(-1)^{k\cdot x}$ are real and every element is its own
inverse, so kernel symmetry is automatic and $\mathcal{K}=\mathcal{K}_+$, which makes the
reflection caveat of Def.~\ref{def:kernels} vacuous, $\hat p(m)$ and $\pg(m)$ are real, and
every envelope is real with $\abs{\hat K(m)}\le1$.

\emph{Filter and kernel.}
Write the low-pass of Sec.~\ref{sec:hypercube} as $g_\eta(k)=\lambda^{\abs{k}}$ with flip rate
$\eta\in[0,\tfrac12]$ and $\lambda=1-2\eta\in[0,1]$. Using
\[
  |k|+|k\oplus m|=|m|+2|k\setminus m|\,,
\]
where $\abs{k\setminus m} = \sum_{j=1}^n k_j(1-m_j)$ counts the 1-bits of $k$ at positions where $m$ has a 0,
the pair weight factorizes as
\[
  W_m(k)
  =g_\eta(k)g_\eta^*(k\oplus m)
  =\lambda^{|m|}\bigl(\lambda^2\bigr)^{|k\setminus m|}.
\]
For the independent bit-flip kernel
\[
  K_\eta(z)=\eta^{|z|}(1-\eta)^{n-|z|},
\]
the Walsh transform (the Fourier transform on the Boolean cube) is
\[
  \hat K_\eta(m)
  =\prod_{j=1}^{n}\big[(1-\eta)+\eta(-1)^{m_j}\big]
  =\lambda^{|m|}.
\]
Therefore
\[
  W_m(k)
  =\hat K_\eta(m)\bigl(\lambda^2\bigr)^{|k\setminus m|},
\]
which is Eq.~\eqref{eq:bitflip_fact}. Thus the pair weight consists of the
classical bit-flip transfer function and an additional $k$-dependent window.

\begin{lemma}\label{lem:real_wm}
On $\mathbb{Z}_2^n$, for every diagonal filter $g$ and every input state,
$S_m=\sum_kW_m(k)\rho_m(k)$ is real. Consequently $\wm\in\mathbb{R}$ and $R_g(m)=0$ at every
regular lag.
\end{lemma}
Indeed $W_0(k)=\abs{g(k)}^2$ and $\rho_0(k)=\abs{\ak}^2$ are real, and for $m\neq0$ the map
$k\mapsto k\oplus m$ is a fixed-point-free involution with $W_m(k\oplus m)=W_m(k)^*$ and
$\rho_m(k\oplus m)=\rho_m(k)^*$, so the sum splits into real pairs
$2\,\mathrm{Re}\!\left[W_m(k)\rho_m(k)\right]$; since $\hat p(m)$ is real so is $\wm$, and
Eq.~\eqref{eq:residual_final_form} gives $R_g(m)=0$. No reality assumption on $g$ enters, so on
the cube the certificate Eq.~\eqref{eq:gap_lower} is empty for phase filters as well as for
magnitude ones, and part~(ii) of Prop.~\ref{prop:mechanisms} carries no information there.\qed

\emph{Matched cost.} With $A_w=\sum_{\abs{k}=w}\abs{\ak}^2$, Eq.~\eqref{eq:psucc} reads
$\psucc(\eta)=\sum_w\lambda^{2w}A_w$, continuous and non-increasing in $\eta$ with $\psucc(0)=1$,
so the affordable flip rate
$\eta^\star(n)=\max\{\eta\in[0,\tfrac12]\colon\psucc(\eta)\ge\pstar\}$ is attained. It is the
counterpart of Eq.~\eqref{eq:sigmastar}, with small $\eta$ playing the role of large $\sigma$.

For a low-degree input, $A_w=0$ for $w>w_0$, one has $\psucc\ge\lambda^{2w_0}$, so any $\eta$ with
$\lambda^{2w_0}\ge\pstar$ is affordable and $\eta^\star=\Omega(1)$ whenever $w_0=O(1)$. Since
$\rho_m(k)\neq0$ requires $\abs{k},\abs{k\oplus m}\le w_0$ while
$\abs{m}\le\abs{k}+\abs{k\oplus m}$, Eqs.~\eqref{eq:wiener_khinchin} and \eqref{eq:master_id}
confine both $\hat p$ and $\pg$ to $\abs{m}\le2w_0$ for every $\eta$. The confinement is exact
rather than a truncation, so the clipping used in the proof of Theorem~\ref{thm:dichotomy} is vacuous: prefix marginals
of $p_g(x)=2^{-n}\sum_m\pg(m)(-1)^{m\cdot x}$ are again Walsh sums of degree at most $2w_0$, so its
$O(n^{2w_0})$ coefficients fix every conditional bit probability and $p_g$ is sampled bit by bit in
$\mathrm{poly}(n)$ time.

For a broadband input, $\sum_{w\le L}A_w\to0$ at every fixed $L$, and splitting $\psucc$ at degree
$L$ gives $\psucc(\eta)\le\sum_{w\le L}A_w+\lambda^{2L}$, the analogue of
Eq.~\eqref{eq:psucc_split}. Affordability then forces $(\lambda^\star)^{2L}\ge\pstar-o(1)$ at every
fixed $L$, hence $\lambda^\star\to1$, i.e.\ $\eta^\star(n)\to0$, and $g_{\eta^\star}(k)\to1$ at
every fixed degree: no fixed degree is attenuated, which is Eq.~\eqref{eq:no_attenuation} read on
the cube. Between the two classes the argument is silent, since a bounded $\eta^\star$ does not by
itself bound the degree of $p_g$.

\section{Washout for spectrally incoherent input}
\label{app:washout}
Consider a magnitude filter $g$ ($W_m$ real,
$W_m(k)=\abs{g(k)}\abs{g(k-m)}\ge0$) and $a_x=\sqrt{p(x)}\,e^{i\varphi_x}$ with $\varphi_x$
independent and uniform. In the following, we show that under these conditions, the spectrum effectively
washes out such that $R_g=0$ in expectation. The coherent residual Eq.~\eqref{eq:coherent_residual} is carried
by $\mathrm{Im}\,\wm$, so it suffices to show $\mathrm{Im}\,\wm\to0$ in mean square as the input spreads.
 Note that $p(x)=\abs{a_x}^2$ is deterministic and the only randomness is
in the phases. To isolate the action of the random phases on the spectrum, we write 
$\wm=S_m/\hat p(m)$ with $S_m\equiv\sum_k W_m(k)\rho_m(k)$ random and $\hat p(m)$ a constant.
With this, the expectation of $\wm$ is
\begin{equation}
  \mathbb E[\wm]=\frac{\mathbb E[S_m]}{\hat p(m)}\,,\qquad
  \mathrm{Var}(\wm)=\frac{\mathrm{Var}(S_m)}{\abs{\hat p(m)}^2}\,.
  \label{eq:expect_wm}
\end{equation}
The case $m=0$ is trivial, since $\pg(0)=\hat p(0)=1$ forces $R_g(0)=0$. In the following, we thus fix a regular lag
$m\neq0$. With the particular form of the amplitudes fixed as above and using Eq.~\eqref{eq:fourier_transf}, we have
\begin{equation}
  \rho_m(k) =\frac{1}{N}\sum_{x,y}\sqrt{p_xp_y}e^{i(\phi_x-\phi_y)}\omega^{-kx+(k-m)y}\,,
\end{equation}
Taking the expectation value over the phases yields 
$\mathbb E[e^{i(\varphi_x-\varphi_y)}]=\delta_{xy}$, such that
\begin{equation}
  \mathbb E[S_m]=\frac{\hat p(m)}{N}\sum_k W_m(k)\,.
  \label{eq:ESm}
\end{equation}
Substituting this into Eq.~\eqref{eq:expect_wm} yields $\mathbb E[\wm]=\sum_k W_m(k)/N$, and since $W_m(k)$ is real by assumption. Consequently,
 $\mathbb{E}[R_g] = 0$ such that a magnitude filter thus produces no residual in expectation and any residual is a fluctuation about the mean whose size we bound next.

Expanding $\mathbb{E}[\abs{S_m}^2]$ yields
\begin{align}
  \mathbb{E}[\abs{S_m}^2]&=\mathbb{E}\left[\sum_{k,k'} W_m(k)W_m(k')\rho_m(k)\rho^*_m(k')\right]\\
  &=\sum_{x,y,x',y'}\sum_{k,k'} W_m(k)W_m(k')\omega^{-kx+(k-m)y+k'x'-(k'-m)y'}\sqrt{p_xp_yp_{x'}p_{y'}}\mathbb{E}\left[\exp(i(\phi_x-\phi_y-\phi_{x'}+\phi_{y'}))\right]\,.\label{eq:expect_ssquared}
\end{align}
Note that the expectation value is nonzero only for the pairings $x=y,\ x'=y'$ and $x=x',\ y=y'$. The first pairing reconstructs
$\abs{\mathbb{E} [S_m]}^2$ and cancels in the variance. Subtracting the coincidence $x=y=x'=y'$ once thus leads to
\begin{align}
  \text{Var}(S_m) &= \frac{1}{N^2}\sum_{x\neq y}\sum_{k,k'} W_m(k)W_m(k')\omega^{-(k-k')(x-y)}p_xp_y\\ 
  &=\frac1{N^2}\sum_{x\neq y}p(x)\,p(y)\,\abs{\widehat W_m(x-y)}^2\,.
  \label{eq:VarSm}
\end{align}
Where we have introduced the Fourier transform of the pair-weights $\hat{W}_m(l)=\sum_k W_m(k)\omega^{-kl}$. Dropping the constraint $x\neq y$ (the omitted terms are nonnegative), bounding one factor
$p(x)\le p_{\max}\equiv \text{max}_x p_x$, completing the sum over $x$ then yields
\begin{equation}
  \text{Var}(S_m)\leq \frac{p_{\rm max}}{N^2}\sum_k\abs{\hat{W}_m(k)}^2\,.
\end{equation}
Using Parseval's theorem Eq.~\eqref{eq:parseval} finally gives
\begin{equation}
  \text{Var}(S_m)\leq \frac{p_{\rm max}}{N}\sum_x W_m(x)^2\leq p_{\rm max}\,.
  \label{eq:sm_bound}
\end{equation}
where we have used $W_m(k)^2=\abs{g(k)}^2\abs{g(k-m)}^2\le1$, so $\sum_k W_m(k)^2\le N$ in the last step. We now return to the residual which in expectation
is bounded as
\begin{equation}
  \mathbb E\big[\abs{R_g(m)}^2\big]
  =\mathbb E\big[(\mathrm{Im}\,\wm)^2\big]\,\frac{\abs{\hat p(m)}^2}{\psucc^2}
  \le\frac{\mathrm{Var}(S_m)}{\psucc^2}
  \le\frac{p_{\max}}{(\pstar)^2}\,,
\end{equation}
which is Eq.~\eqref{eq:washout_bound}. 
Here we have used $\mathbb E\big[(\mathrm{Im}\,\wm)^2\big]\le\mathrm{Var}(\wm)=\mathrm{Var}(S_m)/\abs{\hat p(m)}^2$, which is bounded by Eq.~\eqref{eq:sm_bound}. The bound is uniform in $m$, so by Markov's inequality
$R_g(m)\to0$ in mean-square at every lag as $p_{\max}\to0$ with $\psucc\ge \pstar$ fixed. 

Note that for a phase filter we have $W_m(k)\in\mathbb{C}$. In this case, the residual does not vanish in general
so the same incoherent input does not wash out. \qed

\emph{Scope of the bound.} Equation~\eqref{eq:washout_bound} constrains one lag at a time and
does not transfer to $\Phi_g$. Summed over the $\abs{\mathcal S}=\Theta(N)$ regular lags of a
broadband output it returns at least $1/(\pstar)^2>1$, since $p_{\max}\ge1/N$ always, so it
never feeds the upper bound of Eq.~\eqref{eq:gap_sandwich}. A residual of size $N^{-1/2}$ spread
over every lag, which is what the washout produces, saturates the conversion
$\norm{f}_1\le\sqrt N\norm{f}_2$ behind that bound. What Eq.~\eqref{eq:washout_bound}
establishes is that the certificate Eq.~\eqref{eq:gap_lower} vanishes, not that the gap closes.
The uniform member of Sec.~\ref{sec:washout} separates the two: there
$(p\star K)(x)=\tfrac1N\sum_yK(y)=p(x)$ for every $K\in\mathcal K$, so $\Phi_g=\mathrm{TV}(p_g,p)$
exactly, while matched cost fixes $\sigma^\star/N\to\pstar/\sqrt\pi$, giving the intensity $Np_g$
an $n$-independent nondegenerate law. Numerically, at $\pstar=0.1$ over $32$ phase seeds per
size, the exact gap plateaus at $\Phi_g=0.345,\,0.355,\,0.352,\,0.354,\,0.354$ for $n=8,10,12,14,16$,
while the mode count grows from $35$ to $8684$ and the modulus bound
Eq.~\eqref{eq:real_channel_bound} falls from $0.115$ to $0.013$.

\section{Gaussian example}
\label{app:gauss}

For $g(k)=e^{-k^2/2\sigma^2}$, completing the square in $W_m(k)=g(k)g(k-m)$ gives
$\sup_k|W_m(k)|=e^{-m^2/4\sigma^2}$ (attained at $k=m/2$), so the overlap
profile is $s(u)=e^{-u^2/4}$ and $s(u)^2=e^{-u^2/2}$. Hence
\begin{equation*}
  T(M/\sigma)=\int_{M/\sigma}^\infty e^{-u^2/2}\,du\;\le\;e^{-M^2/2\sigma^2},
\end{equation*}
a Chernoff tail bound (valid for $M\ge\sigma$). Substituting into Eq.~\eqref{eq:T_bound},
\begin{equation*}
  \mathrm{TV}(p_g,p_g^{\le M})\;\le\;\frac{\sqrt{\sigma}\,e^{-M^2/4\sigma^2}}{\sqrt2\,\psucc}.
\end{equation*}
Imposing $\mathrm{TV}\le\epsilon$ and using $\psucc\ge \pstar$ gives, at equality,
\begin{equation*}
  e^{-M^2/2\sigma^2}=\frac{2(\pstar\epsilon)^2}{\sigma}
  \qquad\Longrightarrow\qquad
  M^\star(\sigma;\epsilon)=\sigma\sqrt{2\ln\!\frac{\sigma}{2(\pstar\epsilon)^2}}
  =O\!\big(\sigma\sqrt{\ln(1/\epsilon)}\big),
\end{equation*}
consistent with the general form $M^\star\le\sigma\,T^{-1}\!\big(2(\pstar\epsilon)^2/\sigma\big)$.

\section{Exact gaps for the numerical experiments}\label{app:numerics}
\begin{table}[t]
\caption{Cost of dropping reflection symmetry, by input family, over all instances at
$8\le n\le12$. A dagger marks a deterministic family, with one instance per size.}
\label{tab:kernel_class}
\centering
\begin{tabular}{lccc}
\hline\hline
& & \multicolumn{2}{c}{$\Phi_g^+/\Phi_g$}\\
\cline{3-4}
input & instances & median & min\\
\hline
phase-trivial & $5^\dagger$ & $1.000$ & $1.000$\\
spectrally incoherent & $160$ & $0.997$ & $0.947$\\
coherent chirp & $5^\dagger$ & $0.966$ & $0.966$\\
Gaussian spectrum & $5^\dagger$ & \multicolumn{2}{c}{$\Phi_g=\Phi_g^+=0$}\\
random-parameter circuits & $320$ & $0.931$ & $0.201$\\
trained circuits & $320$ & $0.998$ & $0.833$\\
\hline\hline
\end{tabular}
\end{table}
The gap Eq.~\eqref{eq:gap} is an infimum of a piecewise-linear function of
$K$ over a polytope, hence a linear program:
\begin{equation}
  \begin{aligned}
    \underset{K\in\mathbb{R}^N}{\text{minimize}}\quad
      &\tfrac12\norm{CK-p_g}_1\\[2pt]
    \text{subject to}\quad
      &K(x)\ge0\,,\qquad\sum_xK(x)=1\,,\\[2pt]
      &K(-x)=K(x)\,,
  \end{aligned}
  \label{eq:gap_lp}
\end{equation}
where $C_{xu}=p(x-u)$ is the circulant of the input, so $CK=p\star K$ and the objective is
$\mathrm{TV}(p\star K,p_g)$. The last constraint selects $\mathcal{K}$ and returns $\Phi_g$;
dropping it selects $\mathcal{K}_+$ and returns $\Phi_g^+$.
We solve both programs exactly for every input reported in
Fig.~\ref{fig:gaussian_low_pass}, at $8\le n\le12$ where the $O(N)$ program is affordable:
$1647$ instances in total, comprising the four constructed families and the trained and
random-parameter circuits together with their phase controls. This replaces the bracket of
Eq.~\eqref{eq:gap_lower} by $\Phi_g$ itself, and returns $\Phi_g^+$ alongside it.
The instance counts differ by family only because the families differ in what there is to
sample: the phase-trivial, Gaussian-spectrum and chirp inputs are deterministic constructions
with one state per size, so their five instances are the five sizes and a median or minimum
over them reports the $n$-dependence rather than an ensemble spread; the incoherent family
draws $32$ phase seeds per size, and each circuit family runs $32$ seeds at $d=2$ and $d=4$.
The phase-erased control is the one variant whose program is not affordable at every size, since
its near-uniform filtered output is matched only by a kernel spread over a finite fraction of the
torus; it is solved at $n\le10$, where its gap is already flat
($0.0380,\,0.0380,\,0.0383$ at $n=8,9,10$ for $d=4$), and its instance count is correspondingly
smaller.

Two facts follow. First, dropping reflection symmetry from the kernel family does not result
in meaningful advantages for the classical competitor, which justifies the restriction to
smoothing kernels in the main text (Table~\ref{tab:kernel_class}).
The phase-trivial input has $\Phi_g=\Phi_g^+=0.039$. The Gaussian-spectrum input attains
$\Phi_g=\Phi_g^+=0$ to within the solver residual, so its ratio carries no information and it
is listed separately.
No instance changes cell membership. Every trained circuit keeps $\Phi_g^+\ge0.115$, and the
optimizer over $\mathcal{K}_+$ is strongly asymmetric throughout (median
$\tfrac12\sum_x\abs{K(x)-K(-x)}=0.78$), so the shift freedom is exercised but does not significantly
improve the performance.
There are two exceptions: The chirp loses a fixed $3.4\%$ at every
$n$, which is the deterministic shift of $\mathcal{K}_+$ acting on the displaced output of a
quadratic phase. The random-parameter circuits are the loosest family, with isolated seeds
retaining only $20\%$ of their gap. They remain bounded away from zero
($\min\Phi_g^+=0.036$), but the trained circuits, which carry the claim, do not rely on this.

Second, the bracket is loose downwards and tight upwards: on the trained
circuits, pooled over both depths, the certificate $C_g$ understates $\Phi_g$ by a median factor
$3.0$ (range $1.2$--$34$), whereas the Gaussian-kernel bound exceeds it by a median $0.06\%$
and at most $3.3\%$. The true gap therefore sits at the top of the reported bracket for the
circuits and, to $0.1\%$, for the spectrally incoherent family. The chirp is an exception,
falling near mid-bracket ($\Phi_g=0.515$ against $0.449$ and $0.579$).

Both programs are solved by column generation over the kernel lags with a dual-feasibility
certificate, cross-checked against the full design. Across all instances
$\Phi_g^+-\Phi_g\le2.6\times10^{-10}$, the returned kernel reproduces its own objective to
$3.8\times10^{-8}$, and $\tfrac12\max_m\abs{R_g(m)}\le\Phi_g\le\mathrm{TV}(p\star K_\sigma,p_g)$
holds at every instance.

\end{document}